\documentclass[onecolumn, a4size, 12pt]{IEEEtran}
\usepackage{amsmath}
\usepackage{graphicx}
\usepackage{epsfig}
\usepackage{latexsym}
\usepackage{amsfonts}
\usepackage{amssymb}
\usepackage{psfrag}
\usepackage{cite}
\usepackage{subfigure}
\usepackage{color}
\usepackage{float}

\begin{document}
\pagestyle{plain}
\title{Downlink and Uplink Energy Minimization Through User Association and Beamforming in Cloud RAN
\footnote{This paper has been presented in part at the IEEE International Conference on Acoustics, Speech, and Signal Processing (ICASSP), Florence, Italy, 4-9 May 2014.
}
\footnote{S. Luo and T. J. Lim are with the Department of
Electrical and Computer Engineering, National University of
Singapore (e-mail:\{shixin.luo, eleltj\}@nus.edu.sg).} \footnote{R.
Zhang is with the Department of Electrical and Computer Engineering,
National University of Singapore (e-mail:elezhang@nus.edu.sg). He is
also with the Institute for Infocomm Research, A*STAR, Singapore.}}

\author{Shixin Luo, Rui
Zhang, and Teng Joon Lim}

\setlength{\textwidth}{7.1in} \setlength{\textheight}{9.7in}
\setlength{\topmargin}{-0.8in} \setlength{\oddsidemargin}{-0.30in}

\maketitle

\begin{abstract}
The cloud radio access network (C-RAN) concept, in which densely deployed access points (APs) are empowered by cloud computing to cooperatively support mobile users (MUs), to improve mobile data rates, has been recently proposed. However, the high density of active APs results in severe interference and also inefficient energy consumption. Moreover, the growing popularity of highly interactive applications with stringent uplink (UL) requirements, e.g. network gaming and real-time broadcasting by wireless users, means that the UL transmission is becoming more crucial and requires special attention. Therefore in this paper, we propose a joint downlink (DL) and UL MU-AP association and beamforming design to coordinate interference in the C-RAN for energy minimization, a problem which is shown to be NP hard. Due to the new consideration of UL transmission, it is shown that the two state-of-the-art approaches for finding computationally efficient solutions of joint MU-AP association and beamforming considering only the DL, i.e., group-sparse optimization and relaxed-integer programming, cannot be modified in a straightforward way to solve our problem. Leveraging on the celebrated UL-DL duality result, we show that by establishing a virtual DL transmission for the original UL transmission, the joint DL and UL optimization problem can be converted to an equivalent DL problem in C-RAN with two inter-related subproblems for the original and virtual DL transmissions, respectively. Based on this transformation, two efficient algorithms for joint DL and UL MU-AP association and beamforming design are proposed, whose performances are evaluated and compared with other benchmarking schemes through extensive simulations.
\end{abstract}


\IEEEpeerreviewmaketitle
\setlength{\baselineskip}{1.3\baselineskip}
\newtheorem{definition}{\underline{Definition}}[section]
\newtheorem{fact}{Fact}
\newtheorem{assumption}{Assumption}
\newtheorem{theorem}{\underline{Theorem}}[section]
\newtheorem{lemma}{\underline{Lemma}}[section]
\newtheorem{corollary}{Corollary}
\newtheorem{proposition}{\underline{Proposition}}[section]
\newtheorem{example}{\underline{Example}}[section]
\newtheorem{remark}{\underline{Remark}}[section]
\newtheorem{algorithm}{\underline{Algorithm}}[section]
\newcommand{\mv}[1]{\mbox{\boldmath{$ #1 $}}}

\section{Introduction}\label{sec:introduction}
To meet the fast growing mobile data volume driven by applications such as smartphones and tablets, the traditional wireless network architecture based on a single layer of macro-cells has shifted to one composed of smaller cells such as pico/femto cells with more densely deployed access points (APs). Therefore, cloud radio access network (C-RAN) \cite{CRAN} has recently been proposed and drawn a great deal of attention. In a C-RAN, the distributed APs, also termed remote radio heads (RRHs), are connected to the baseband unit (BBU) pool through high bandwidth backhaul links, e.g. optical transport network \cite{optical}, to enable centralized processing, collaborative transmission, and real-time cloud computing. As a result, significant rate improvement can be achieved due to reduced pathloss along with joint scheduling and signal processing.

{\color{red}However, with densely deployed APs, several new challenges arise in C-RAN. First, close proximity of many active APs results in increased interference, and hence the transmit power of APs and/or mobile users (MUs) needs to be increased to meet any given quality of service (QoS). Second, the amount of energy consumed by a large number of active APs \cite{3GppTR} as well as by the transport network to support high-capacity connections with the BBU pool \cite{backhaul} will also become considerable.} Such facts motivate us to optimize the energy consumption in C-RAN, which is the primary concern of this paper. In particular, both downlink (DL) and uplink (UL) transmissions are considered jointly. The studied C-RAN model consists of densely deployed APs jointly serving a set of distributed MUs, where CoMP based joint transmit/receive processing (beamforming) over all active APs is employed for DL/UL transmissions. Under this setup, we study a joint DL and UL MU-AP association and beamforming design problem to minimize the total energy consumption in the network subject to MUs' given DL and UL QoS requirements. The energy saving is achieved by optimally assigning MUs to be served by the minimal subset of active APs, finding the power levels to transmit at all MUs and APs, and finding the beamforming vectors to use at the multi-antenna APs.

{\color{red}This problem has not been investigated to date, and the closest prior studies are \cite{TLuo13, Letaief13, Liao13, Yu13, Cheng13}. However, the prior studies have all considered MU association and/or active AP selection problems for various objectives from the DL perspective. In particular, the problems studied in \cite{Letaief13, Liao13, Cheng13} can be treated as the DL-only version of our considered joint DL and UL problem, in which the transmit beamforming vectors and the active set of APs are jointly optimized to minimize the power consumption at all the APs. Note that the MU association and/or active AP selection based on DL only may result in inefficient transmit power of MUs or even their infeasible transmit power in the UL considering various possible asymmetries between the DL and UL in terms of channel, traffic and hardware limitation. Furthermore, with users increasingly using applications with high-bandwidth UL requirements, UL transmission is becoming more important. For example, the upload speed required for full high definition (HD) $1080$p Skype video calling is about $20$ Mbps \cite{Skype}. Therefore, we need to account for both DL and UL transmissions while designing the MUs association and active AP selection scheme. The UL-only MU association problem has also been considered extensively in the literature \cite{Hanly95, Rashid98, Hong13}; however, their solutions are not applicable in the context of this work due to their assumption of one-to-one MU-AP association. It is worth noting that the joint MUs association and active AP selection is mathematically analogous to the problem of antenna selection in large multiple-input multiple-output (MIMO) systems \cite{Mehanna13}, which aims to reduce the number of radio transmission chains and hence the energy consumption and signal processing overhead. The connection between these two problems can be recognized by treating the C-RAN as a distributed large MIMO system.}

In terms of other related work, there have been many attempts to optimize the energy consumption in cellular networks, but only over a single dimension each time, e.g. power control \cite{Zander11}, AP ``on/off'' control \cite{Niu10,shixin12,son11}, and coordinated multi-point (CoMP) transmission \cite{Weiyu10, Zluo11}. To avoid an infeasible power allocation, it was suggested in \cite{Zander11} to gradually remove the MUs that cannot be supported due to their limited transmit power budgets. In addition to achieving energy saving from MUs' perspective, \cite{Niu10,shixin12,son11} proposed to switch off the APs that are under light load to save energy by exploiting the fact that the traffic load in cellular networks fluctuates substantially over both space and time due to user mobility and traffic burstiness. Cooperation among different cells or APs could be another possible way to mitigate the interference and achieve energy-efficient communication. For example, if a certain cluster of APs can jointly support all the MUs, the intercell interference can be further reduced especially for the cell-edge MUs \cite{Weiyu10, Zluo11}. A judicious combination of these techniques should provide the best solution, and this is the direction of our work.

Unfortunately, the considered joint DL and UL MU-AP association and beamforming design problem in this paper involves integer programming and {\color{red}is NP hard as shown for a similar problem in \cite[Theorem 1]{Liao13}}. To tackle this difficulty, two different approaches, i.e., group-sparse optimization (GSO) and relaxed-integer programming (RIP), have been adopted in \cite{Letaief13, Liao13} and \cite{Cheng13}, respectively, to solve a similar DL-only problem, where two polynomial-time algorithms were proposed and shown to achieve good performance through simulations. In particular, the GSO approach is motivated by the fact that in the C-RAN with densely deployed APs, only a small fraction of the total number of APs needs to be active for meeting all MUs' QoS. However, due to the new consideration of UL transmission in this paper, we will show that the algorithms proposed in \cite{Letaief13, Liao13, Cheng13} cannot be applied directly to solve our problem, and therefore the methods derived in this paper are important advances in this field.

{\color{red}The contributions of this paper are summarized as follows:
\begin{enumerate}
\item To optimize the energy consumption tradeoffs between the active APs and MUs, we jointly study the DL and UL MU-AP association and beamforming design by solving a weighted sum-power minimization problem. To our best knowledge, this paper is the first attempt to unify the DL and UL MU-AP association and beamforming design into one general framework.
\item Due to a critical scaling issue in the UL receive beamforming design, the GSO based algorithm \cite{Letaief13, Liao13} and the RIP based algorithm \cite{Cheng13} cannot be applied to solve our joint DL and UL design directly. To address this issue, we establish a virtual DL transmission for the original UL transmission in C-RAN by first ignoring the individual (per-AP and per-MU) power constraints based on the celebrated UL-DL duality result \cite{Boche05}. Consequently, the considered joint DL and UL problem without individual power constraints can be transformed into an equivalent DL problem with two inter-related subproblems corresponding to the original and virtual DL transmissions, respectively. With the equivalent DL-only formulation, we extend the GSO based and RIP based algorithms to solve the relaxed joint DL and UL optimization problem.
\item Considering the fact that the optimal solution to the UL sum-power minimization is component-wise minimum, there is no tradeoff among different MUs in terms of power minimization in the UL. Consequently, we are not able to establish the duality result for the per-MU power constraints in the UL, which are thus difficult to incorporate into our developed algorithm. To resolve this issue, we propose a price based iterative method to further optimize the set of active APs while satisfying the per-MU power constraints. Finally, we verify the effectiveness of our proposed algorithms by extensive simulations from three perspectives: ensuring feasibility for both DL and UL transmissions; achieving optimal network energy saving with MU-AP association; and flexibly adjusting various power consumption tradeoffs between active APs and MUs.
\end{enumerate}}

{\color{red}It is worth pointing out that as the baseband processing is migrated to a central unit, i.e., BBU pool, the data exchanged between the APs and the BBU pool includes oversampled real-time digital signals with very high bit rates (in the order of Gbps). As a result, the capacity requirement for the backhaul/fronthaul links becomes far more stringent in the C-RAN. Given finite backhaul capacity, the optimal strategy for backhaul compression and quantization has been studied recently in e.g. \cite{robust, joint, Quek13}. In this paper, however, we focus on addressing the energy consumption (including both transmission and non-transmission related portions) issue in the C-RAN, which is also one of the major concerns for future cellular networks, by assuming that the backhaul transport network is provisioned with sufficiently large capacity. Note that the optical network has been widely accepted as a good option to implement the high-bandwidth backhaul transport network \cite{optical}.}

The rest of this paper is organized as follows. Section \ref{sec:system model} introduces the C-RAN model, and the power consumption models for the APs and MUs. Section \ref{sec:problem formulation} presents our problem formulation, introduces the two existing approaches, namely GSO and RIP, and explain the new challenges in solving the joint DL and UL optimization. Section \ref{sec:joint UL DL} presents our proposed algorithms based on the virtual DL representation of the UL transmission. Section \ref{sec:numerical} shows numerical results. Finally, Section \ref{sec:conclusion} concludes the paper.

{\it Notations}: Boldface letters refer to vectors (lower case) or matrices (upper case). For an arbitrary-size matrix $\mathbf{M}$, $\mathbf{M}^{*}$, $\mathbf{M}^{H}$, and $\mathbf{M}^{T}$ denote the complex conjugate, conjugate transpose and transpose of $\mathbf{M}$, respectively. The distribution of a circularly
symmetric complex Gaussian (CSCG) random vector with mean vector $\mathbf{x}$ and covariance matrix $\boldsymbol\Sigma$ is denoted by
$\mathcal{CN}(\mathbf{x},\boldsymbol\Sigma)$; and $\thicksim$ stands for ``distributed as''. $\mathbb{C}^{x\times y}$ denotes the space of $x\times y$ complex matrices. $\|\mathbf{x}\|$ denotes the Euclidean norm of a complex vector $\mathbf{x}$, and $|z|$ denotes the magnitude of a complex number $z$.

\section{System Model}\label{sec:system model}
We consider a densely deployed C-RAN \cite{CRAN, DenseHetnet} consisting of $N$ access points (APs), denoted by the set $\mathcal{N} = \{1,\cdots,N\}$. The set of distributed APs jointly support $K$ randomly located mobile users (MUs), denoted by the set $\mathcal{K} = \{1,\cdots,K\}$, for both downlink (DL) and uplink (UL) communications. In this paper, for the purpose of exposition, we consider linear precoding and decoding in the DL and UL, respectively, which is jointly designed at the BBU pool assuming the perfect channel knowledge for all MUs. {\color{red}The results in this paper can be readily extended to the case of more complex successive precoding/decoding, e.g. dirty-paper coding (DPC) \cite{dirty} and multiuser detection with successive interference cancelation (SIC) \cite{Tse}, with fixed coding orders among the users.} We also assume that each AP $n$, $n \in \mathcal{N}$, is equipped with $M_n \geq 1$ antennas, and all MUs are each equipped with one antenna. It is further assumed that there exist ideal low-latency backhaul transport links with sufficiently large capacity (e.g. optical fiber) connecting the set of APs to the BBU pool, which performs all the baseband signal processing and transmission scheduling for all APs. The centralized architecture results in efficient coordination of the transmission/reception among all the APs, which can also be opportunistically utilized depending on the traffic demand.

We consider a quasi-static fading environment, and denote the channel vector in the DL from AP $n$ to MU $i$ and that in the UL from MU $i$ to AP $n$ as $\mathbf{h}^{H}_{i,n} \in \mathbb{C}^{1\times M_n}$ and $\mathbf{g}_{i,n} \in \mathbb{C}^{M_n \times 1}$, respectively. Let the vector consisting of the channels from all the APs to MU $i$ and that consisting of the channels from MU $i$ to all the APs be $\mathbf{h}^{H}_{i} = \left[\mathbf{h}^{H}_{i,1}, \cdots ,\mathbf{h}^{H}_{i,N}\right]$ and $\mathbf{g}_{i} = \left[\mathbf{g}^{T}_{i,1}, \cdots ,\mathbf{g}^{T}_{i,N}\right]^{T}$, respectively. There are two main techniques for separating DL and UL transmissions on the same physical transmission medium, i.e., time-division duplex (TDD) and frequency-division duplex (FDD). If TDD is assumed, channel reciprocity is generally assumed to hold between DL and UL transmissions, which means that the channel vector $\mathbf{g}_{i}$ in the UL is merely the transpose of that $\mathbf{h}^{H}_i$ in the DL, i.e., $\mathbf{g}_{i} = \mathbf{h}^{*}_i, \forall i \in \mathcal{K}$. However, if FDD is assumed, $\mathbf{h}_i$'s and $\mathbf{g}_{i}$'s are different in general.

\subsection{DL Transmission}
In DL transmission, the transmitted signal from all APs can be generally expressed as
\begin{align}\label{eq:DL transmit signal from all APs}
\mathbf{x}^{\text{DL}} = \sum^{K}_{i=1}\mathbf{w}^{\text{DL}}_is^{\text{DL}}_i
\end{align}
where $\mathbf{w}^{\text{DL}}_i \in \mathbb{C}^{M\times1}$ is the beamforming vector for all APs to cooperatively send one single stream of data signal $s_i^{\text{DL}}$ to MU $i$, which is assumed to be a complex random variable with zero mean and unit variance. Note that $\sum^{N}_{n=1}M_n = M$. Then, the transmitted signal from AP $n$ can be expressed as
\begin{align}\label{eq:DL transmit signal from AP n}
\mathbf{x}^{\text{DL}}_n = \sum^{K}_{i=1}\mathbf{w}^{\text{DL}}_{i,n}s^{\text{DL}}_i, ~~ n = 1,\cdots,N
\end{align}
where $\mathbf{w}^{\text{DL}}_{i,n} \in \mathbb{C}^{M_n\times1}$ is the $n$th block component of $\mathbf{w}^{\text{DL}}_i$, corresponding to the transmit beamforming vector at AP $n$ for MU $i$. Note that $\mathbf{x}^{\text{DL}} = [\left(\mathbf{x}^{\text{DL}}_1\right)^{T},\cdots,\left(\mathbf{x}^{\text{DL}}_N\right)^{T}]^{T}$ and $\mathbf{w}^{\text{DL}}_i = [\left(\mathbf{w}^{\text{DL}}_{i,1}\right)^{T},\cdots,\left(\mathbf{w}^{\text{DL}}_{i,N}\right)^{T}]^{T}$, $i=1,\cdots,K$. From (\ref{eq:DL transmit signal from AP n}), the transmit power of AP $n$ in DL is obtained as
\begin{align}\label{eq:DL transmit power at AP n}
p^{\text{DL}}_n = \sum^{K}_{i=1}\|\mathbf{w}^{\text{DL}}_{i,n}\|^2, ~~ n = 1,\cdots,N.
\end{align}
We assume that there exists a maximum transmit power constraint for each AP $n$, i.e.,
\begin{align}\label{eq:per AP power constraint}
p^{\text{DL}}_n \leq P^{\text{DL}}_{n,\text{max}}, ~~ n = 1,\cdots,N.
\end{align}

The received signal at the $i$th MU is then expressed as
\begin{align}\label{eq:DL receive signal at MU i}
y_i^{\text{DL}} = \mathbf{h}_i^{H}\mathbf{w}^{\text{DL}}_i s_i^{\text{DL}} + \sum^{K}_{j\neq i}\mathbf{h}_i^{H}\mathbf{w}^{\text{DL}}_j s_j^{\text{DL}} + z^{\text{DL}}_i, ~~ i = 1, \cdots, K
\end{align}
where $z^{\text{DL}}_i$ is the receiver noise at MU $i$, which is assumed to be a circularly symmetric complex Gaussian (CSCG) random variable with zero mean and variance $\sigma^2$, denoted by $z^{\text{DL}}_i \thicksim \mathcal{CN}(0, \sigma^2)$. Treating the interference as noise, the signal-to-interference-plus-noise ratio (SINR) in DL for MU $i$ is given by
\begin{align}\label{eq:DL SINR}
\text{SINR}^{\text{DL}}_i =  \frac{|\mathbf{h}_i^{H}\mathbf{w}^{\text{DL}}_i|^2}{\sum_{j \neq i}|\mathbf{h}_i^{H}\mathbf{w}^{\text{DL}}_j|^2 + \sigma^2}, ~~ i = 1,\cdots, K.
\end{align}

\subsection{UL Transmission}
In UL transmission, the transmitted signal from MU $i$ is given by
\begin{align}\label{eq:UL transmit signal from MU i}
x^{\text{UL}}_i = \sqrt{p^{\text{UL}}_i}s^{\text{UL}}_i, ~~ i =1,\cdots, K
\end{align}
where $p^{\text{UL}}_i$ denotes the transmit power of MU $i$, and $s^{\text{UL}}_i$ is the information bearing signal which is assumed to be a complex random variable with zero mean and unit variance. With $P^{\text{UL}}_{i,\text{max}}$ denoting the transmit power limit for MT $i$, it follows that
\begin{align}
p^{\text{UL}}_i \leq P^{\text{UL}}_{i,\text{max}}, ~~ i =1,\cdots, K.
\end{align}

The received signal at all APs is then expressed as
\begin{align}\label{eq:UL receive signal at all APs}
\mathbf{y}^{\text{UL}} = \sum^{K}_{i=1}\mathbf{g}_i\sqrt{p^{\text{UL}}_i}s^{\text{UL}}_i + \mathbf{z}^{\text{UL}}
\end{align}
where $\mathbf{z}^{\text{UL}} \in \mathbb{C}^{M\times1}$ denotes the receiver noise vector at all APs consisting of independent CSCG random variables each distributed as $\mathcal{CN}\left(0,\sigma^2\right)$. Let $\mathbf{v}^{\text{UL}}_i \in \mathbb{C}^{M\times1}$ denote the receiver beamforming vector used to decode $s_i^{\text{UL}}$ from MU $i$. Then the SINR in UL for MU $i$ after applying $\mathbf{v}^{\text{UL}}_i$ is given by
\begin{align}\label{eq:UL SINR}
\text{SINR}^{\text{UL}}_i = \frac{p^{\text{UL}}_i|(\mathbf{v}^{\text{UL}}_i)^{T}\mathbf{g}_i|^2}{\sum_{j \neq i}p^{\text{UL}}_j|(\mathbf{v}^{\text{UL}}_i)^{T}\mathbf{g}_j|^2 + \sigma^2\|\mathbf{v}^{\text{UL}}_i\|^2}, ~~ i =1,\cdots, K.
\end{align}
Let $\mathbf{v}^{\text{UL}}_{i,n} \in \mathbb{C}^{M_n\times1}$ denote the $n$th block component in $\mathbf{v}^{\text{UL}}_i$, corresponding to the receive beamforming vector at AP $n$ for MU $i$. We thus have $\mathbf{v}^{\text{UL}}_i = [\left(\mathbf{v}^{\text{UL}}_{i,1}\right)^{T},\cdots,\left(\mathbf{v}^{\text{UL}}_{i,N}\right)^{T}]^{T}$, $i=1,\cdots,K$.

\subsection{Energy Consumption Model}
The total energy consumption in the C-RAN comprises of the energy consumed by all APs and all MUs. From (\ref{eq:DL transmit signal from all APs}) and (\ref{eq:UL transmit signal from MU i}), the total transmit power of all APs in DL and that of all MUs in UL can be expressed as $P_t^{\text{DL}} = \sum^{K}_{i=1}\|\mathbf{w}^{\text{DL}}_{i}\|^2$ and $P_t^{\text{UL}} = \sum^{K}_{i=1}p^{\text{UL}}_i$, respectively.

Besides the static power consumption at each AP $n$ due to e.g. real-time A/D and D/A processing, denoted as $P_{s,n}, \forall n \in \mathcal{N}$, in C-RAN with centralized processing, the extensive use of high-capacity backhaul links to connect all APs with the BBU pool makes the power consumption of the transport network no more negligible \cite{backhaul}. For example, consider the passive optical network (PON) to implement the backhaul transport network \cite{optical}. The PON assigns an optical line terminal (OLT) to connect to a set of associated optical network units (ONUs), which coordinate the set of transport links connecting all the APs to the BBU pool, each through a single fiber. For simplicity, the resulting power consumption in the PON can be modeled as \cite{optical}
\begin{align}
P_{\text{PON}} = P_{\text{OLT}} + \sum^{N}_{n=1}P_{\text{ONU},n}
\end{align}
where $P_{\text{OLT}}$ and $P_{\text{ONU},n}$ are both constant and denote the power consumed by the OLT and the transport link associated with AP $n$, respectively.

Moreover, we consider that for energy saving, some APs and their associated transport links can be switched into {\color{red}sleep mode \cite{optical, Blume10} (compared with active mode) with negligible power consumption\footnote{\color{red}It is assumed that when the AP is in the sleep mode, it acts as a passive node and listens to the pilot signals transmitted from the MUs for channel estimation, which consumes negligible power compared with being in the active mode for data transmission. It is further assumed that each AP can switch between the active and sleep modes frequently.}}; thus, the total static power consumption of AP $n$, denoted by $P_{c,n} = P_{s,n} + P_{\text{ONU},n}$, $n \in \mathcal{N}$, can be saved if AP $n$ and its associated transport link are switched into sleep mode for both transmission in DL and UL. For convenience, we express the total static power consumption of all active APs as
\begin{align}\label{eq:total circuit power}
P_c = \sum^{N}_{n=1}\mathbf{1}_n\left(\{\mathbf{w}^{\text{DL}}_{i,n}\},\{\mathbf{v}^{\text{UL}}_{i,n}\}\right)P_{c,n}
\end{align}
where $\mathbf{1}_n\left(\cdot\right)$, $n \in \mathcal{N}$, is an indicator function for AP $n$, which is defined as
\begin{align}\label{eq:on off condition}
\mathbf{1}_n\left(\{\mathbf{w}^{\text{DL}}_{i,n}\},\{\mathbf{v}^{\text{UL}}_{i,n}\}\right) & = \left\{ \begin{array}{cl} \displaystyle
0 & \mbox{if } \mathbf{w}^{\text{DL}}_{i,n} = \mathbf{v}^{\text{UL}}_{i,n} = \mathbf{0}, \forall i \in \mathcal{K} \\
1 & \mbox{otherwise. }
\end{array}\right.
\end{align}
Note that in practical PON systems, the OLT in general cannot be switched into sleep mode as it plays the role of distributor, arbitrator, and aggregator of the transport network, which has a typical fixed power consumption of $P_{\text{OLT}} = 20$W \cite{optical}. We thus ignore $P_{\text{OLT}}$ since it is only a constant. From (\ref{eq:on off condition}), MU $i$ is associated with an active AP $n$ if its corresponding transmit and/or receive beamforming vector at AP $n$ is nonzero, i.e., $\mathbf{w}^{\text{DL}}_{i,n} \neq \mathbf{0}$ and/or $\mathbf{v}^{\text{UL}}_{i,n} \neq \mathbf{0}$. Under this setup, it is worth pointing out that each MU $i$ is allowed to connect with two different sets of APs for DL and UL transmissions, respectively, e.g.  $\mathbf{w}^{\text{DL}}_{i,n} \neq \mathbf{0}$ but $\mathbf{v}^{\text{UL}}_{i,n} = \mathbf{0}$ for some $n \in \mathcal{N}$, which is promising to be implemented in next generation cellular networks \cite{Sachs13}. Furthermore, from (\ref{eq:on off condition}), AP $n$ could be switched into sleep mode only if it does not serve any MU.

We aim to minimize the total energy consumption in the C-RAN, including that due to transmit power of all MUs (but ignoring any static power consumption of MU terminals) as well as that due to transmit power and static power of all active APs. Therefore, we consider the following weighted sum-power as our design metric:
\begin{align}\label{eq:weighted sum power}
P_{\text{total}}\left(\{\mathbf{w}^{\text{DL}}_{i}\},\{\mathbf{v}^{\text{UL}}_{i}\}\right) & = \left( \sum^{N}_{n=1}\mathbf{1}_n\left(\{\mathbf{w}^{\text{DL}}_{i,n}\},\{\mathbf{v}^{\text{UL}}_{i,n}\}\right)P_{c,n}\right. \nonumber \\ &\left.+\sum^{K}_{i=1}\|\mathbf{w}^{\text{DL}}_{i}\|^2\right)+\lambda\left(\sum^{K}_{i=1}p^{\text{UL}}_i\right)
\end{align}
where $\lambda \geq 0$ is a weight to trade off between the total energy consumptions between all the active APs and all MUs.

\section{Problem Formulation and Two Solution Approaches}\label{sec:problem formulation}
To minimize the weighted power consumption in (\ref{eq:weighted sum power}), we jointly optimize the DL and UL MU-AP association and transmit/receive beamforming by considering the following problem.
\begin{align}
(\text{P1}):\mathop{\mathtt{Min.}}\limits_{\{\mathbf{w}^{\text{DL}}_i\},\{\mathbf{v}^{\text{UL}}_i\},\{p^{\text{UL}}_i\}} &
~~  P_{\text{total}}\left(\{\mathbf{w}^{\text{DL}}_{i}\},\{\mathbf{v}^{\text{UL}}_{i}\}\right) \label{eq:ULDL objective function}\\
\mathtt{s.t.}
& ~~ \text{SINR}^{\text{DL}}_i \geq \gamma^{\text{DL}}_i, \forall i \in \mathcal{K} \label{eq:ULDL DL SINR constraint}\\
& ~~ \text{SINR}^{\text{UL}}_i \geq \gamma^{\text{UL}}_i, \forall i \in \mathcal{K} \label{eq:ULDL UL SINR constraint}\\
& ~~ p^{\text{DL}}_n \leq P^{\text{DL}}_{n,\text{max}}, ~~ \forall n \in \mathcal{N} \label{eq:ULDL DL power constraint} \\
& ~~ 0 \leq p^{\text{UL}}_i \leq P^{\text{UL}}_{i,\text{max}}, \forall i \in \mathcal{K} \label{eq:ULDL UL power constraint}
\end{align}
where $\gamma^{\text{DL}}_i$ and $\gamma^{\text{UL}}_i$ are the given SINR requirements of MU $i$ for the DL and UL transmissions, respectively. In the rest of this paper, the constraints in (\ref{eq:ULDL DL power constraint}) and (\ref{eq:ULDL UL power constraint}) are termed per-AP and per-MU power constraints, respectively. Problem (P1) can be shown to be non-convex due to the implicit integer programming involved due to indicator function $\mathbf{1}_n\left(\cdot\right)$'s in the objective. Prior to solving problem (P1), we first need to check its feasibility. Since the DL and UL transmissions are coupled only by the objective function in (\ref{eq:ULDL objective function}), the feasibility of problem (P1) can be checked by considering two separate feasibility problems: one for the DL and the other for the UL, which have both been well studied in the literature \cite{Liu12} and thus the details are omitted here for brevity. For the rest of this paper, we assume that problem (P1) is always feasible if all APs are active.

As mentioned in Section \ref{sec:introduction}, the problem of joint MU-AP association and transmit beamforming subject to MUs' QoS and per-AP power constraints for power minimization in the DL-only transmission has been recently studied in \cite{Letaief13, Liao13, Cheng13} using the approaches of GSO and RIP, respectively, where two different polynomial-time algorithms were proposed and shown to both achieve good performance by simulations. In contrast, problem (P1) in this paper considers both DL and UL transmissions to address possible asymmetries between the DL and UL in terms of channel realization, traffic load and hardware limitation. Furthermore, considering that MUs are usually powered by finite-capacity batteries as compared to APs that are in general powered by the electricity grid, we study the power consumption tradeoffs between APs and MUs by minimizing the weighted sum-power $P_{\text{total}}\left(\{\mathbf{w}^{\text{DL}}_{i}\},\{\mathbf{v}^{\text{UL}}_{i}\}\right)$ in (P1). Therefore, the problems considered in \cite{Letaief13, Liao13, Cheng13} can be treated as special cases of (P1).

In the following, we show that due to the new consideration of UL transmission, the algorithms proposed in \cite{Letaief13, Liao13, Cheng13} based on GSO and RIP for solving the DL optimization cannot be applied directly to solve (P1), which thus motivates us to find a new method to resolve this issue in Section \ref{sec:joint UL DL}.

\subsection{GSO based Solution}
Given the fact that the static power, i.e., $P_{c,n}$, is in practice significantly larger than the transmit power at each AP $n$, to minimize the total network energy consumption \cite{optical, wireless_network_power}, it is conceivable that for the optimal solution of (P1) only a subset of $N$ APs should be active. As a result, a ``group-sparse'' property can be inferred from the following concatenated beamforming vector:
\begin{align}\label{eq:group sparsity}
\left[[\hat{\mathbf{w}}^{\text{DL}}_1,\hat{\mathbf{v}}^{\text{UL}}_1],\cdots,[\hat{\mathbf{w}}^{\text{DL}}_N,\hat{\mathbf{v}}^{\text{UL}}_N]\right]
\end{align}
in which $\hat{\mathbf{w}}^{\text{DL}}_n = \left[(\mathbf{w}^{\text{DL}}_{1,n})^T,\cdots,(\mathbf{w}^{\text{DL}}_{K,n})^T\right]$ and $\hat{\mathbf{v}}^{\text{UL}}_n = \left[(\mathbf{v}^{\text{UL}}_{1,n})^T,\cdots,(\mathbf{v}^{\text{UL}}_{K,n})^T\right]$, $n = 1,\cdots,N$, i.e., the beamforming vectors are grouped according to their associated APs . If AP $n$ is in the sleep mode, its corresponding block $[\hat{\mathbf{w}}^{\text{DL}}_n,\hat{\mathbf{v}}^{\text{UL}}_n]$ in (\ref{eq:group sparsity}) needs to be zero. Consequently, the fact that a small subset of deployed APs is selected to be active implies that the concatenated beamforming vector in (\ref{eq:group sparsity}) should contain only a very few non-zero block components.

One well-known approach to enforce desired group sparsity in the obtained solutions for optimization problems is by adding to the objective function an appropriate penalty term. The widely used group sparsity enforcing penalty function, which was first introduced in the context of the group least-absolute selection and shrinkage operator (LASSO) problem \cite{Yuan06}, is the mixed $\ell_{1,2}$ norm. In our case, such a penalty is expressed as
\begin{align}\label{eq:ULDL penalty}
\sum^{N}_{n=1}\left\|[\hat{\mathbf{w}}^{\text{DL}}_n,\hat{\mathbf{v}}^{\text{UL}}_n]\right\|.
\end{align}
The $\ell_{1,2}$ norm in (\ref{eq:ULDL penalty}), similar to $\ell_{1}$ norm, offers the closest convex approximation to the $\ell_{0}$ norm over the vector consisting of $\ell_{2}$ norms $\left\{\left\|[\hat{\mathbf{w}}^{\text{DL}}_n,\hat{\mathbf{v}}^{\text{UL}}_n]\right\|\right\}^{N}_{n=1}$, implying that each $\left\|[\hat{\mathbf{w}}^{\text{DL}}_n,\hat{\mathbf{v}}^{\text{UL}}_n]\right\|$ is desired to be set to zero to obtain group sparsity.

More generally, the mixed $\ell_{1,p}$ norm has also been shown to be able to recover group sparsity with $p > 1$ \cite{GroupSparse}, among which the $\ell_{1,\infty}$ norm, defined as
\begin{align}\label{eq:infinity}
\sum^{N}_{n=1}\max\left(\max\limits_{i,j}\left|w^{\text{DL}}_{i,n}(j)\right|, \max\limits_{i,j}\left|v^{\text{UL}}_{i,n}(j)\right|\right)
\end{align}
has been widely used \cite{Mehanna13}. Compared with $\ell_{1,2}$ norm, $\ell_{1,\infty}$ norm has the potential to obtain more sparse solution but may lead to undesired solution with components of equal magnitude. In this paper, we focus on the $\ell_{1,2}$ norm in (\ref{eq:ULDL penalty}) for our study. We will compare the performance of $\ell_{1,2}$ and $\ell_{1,\infty}$ norms by simulations in Section \ref{sec:numerical}.

According to \cite{Letaief13, Liao13, TLuo13}, at first glance it seems that using the $\ell_{1,2}$ norm, problem (P1) can be approximately solved by replacing the objective function with
\begin{align}\label{eq:scaling problem objective}
\sum^{N}_{n=1}\beta_n\sqrt{\sum^{K}_{i=1}\|\mathbf{w}^{\text{DL}}_{i,n}\|^2+\|\mathbf{v}^{\text{UL}}_{i,n}\|^2} + \sum_{i=1}^{K}\|\mathbf{w}^{\text{DL}}_i\|^2 + \lambda\sum_{i=1}^{K}p^{\text{UL}}_i
\end{align}
where $\sum^{N}_{n=1}\beta_n\sqrt{\sum^{K}_{i=1}\|\mathbf{w}^{\text{DL}}_{i,n}\|^2+\|\mathbf{v}^{\text{UL}}_{i,n}\|^2}$ can be treated as a convex relaxation of the indicator functions in (\ref{eq:weighted sum power}), and $\beta_n \geq 0$ indicates the relative importance of the penalty term associated with AP $n$. However, problem (P1) with (\ref{eq:scaling problem objective}) as the objective function is still non-convex due to the constraints in (\ref{eq:ULDL DL SINR constraint}) and (\ref{eq:ULDL UL SINR constraint}). Furthermore, since the UL receive beamforming vector $\mathbf{v}^{\text{UL}}_i$'s can be scaled down to be arbitrarily small without affecting the UL SINR defined in (\ref{eq:UL SINR}), minimizing (\ref{eq:scaling problem objective}) directly will result in all $\mathbf{v}^{\text{UL}}_i$'s going to zero. To be more specific, let $\hat{\mathbf{w}}^{\text{DL}}_i$ and $\hat{\mathbf{v}}^{\text{UL}}_i$ denote the optimal solution of problem (P1) with (\ref{eq:scaling problem objective}) as the objective function. Then, it follows that
\begin{align}
\hat{\mathbf{v}}^{\text{UL}}_i \approx \mathbf{0}, ~~ \forall i \in \mathcal{K}
\end{align}
and $\hat{\mathbf{w}}^{\text{DL}}_i, \forall i \in \mathcal{K}$, preserves the ``group-sparse'' property where the non-zero block components correspond to the active APs. Two issues thus arise: first, the UL does not contribute to the selection of active APs; second, the set of selected active APs based on the DL only cannot guarantee the QoS requirements for the UL. As a result, the $\ell_{1,2}$ norm penalty term in (\ref{eq:scaling problem objective}) or more generally the $\ell_{1,p}$ norm penalty does not work for the joint DL and UL AP selection in our problem, and hence the algorithm proposed in \cite{Letaief13,Liao13,TLuo13}, which involves only the DL transmit beamforming vector $\mathbf{w}^{\text{DL}}_i$'s, cannot be modified in a straightforward way to solve our problem.

\subsection{RIP based Solution}
Next, we reformulate problem (P1) by introducing a set of binary variable $\rho_n$'s indicating the ``active/sleep'' state of each AP as follows.
\begin{align}
(\text{P2}):\mathop{\mathtt{Min.}}\limits_{\{\mathbf{w}^{\text{DL}}_i\},\{\mathbf{v}^{\text{UL}}_i\},\{p^{\text{UL}}_i\},\{\rho_n\}} &
~~ \left(\sum^{N}_{n = 1}\rho_nP_{c,n}+\sum^{K}_{i=1}\|\mathbf{w}^{\text{DL}}_{i}\|^2\right)+\lambda\left(\sum^{K}_{i=1}p^{\text{UL}}_i\right) \label{eq:ULDL objective function integer}\\
\mathtt{s.t.}
& ~~ (\ref{eq:ULDL DL SINR constraint}), (\ref{eq:ULDL UL SINR constraint}), (\ref{eq:ULDL DL power constraint}), \text{and} ~ (\ref{eq:ULDL UL power constraint}) \\
& ~~ \sum^{K}_{i=1}\|\mathbf{w}^{\text{DL}}_{i,n}\|^2+\|\mathbf{v}^{\text{UL}}_{i,n}\|^2 \leq \rho_n(P^{\text{DL}}_{n,\text{max}}+\eta), \forall n \in \mathcal{N} \label{eq:bigM onoff}\\
& ~~ \rho_n \in \{0,1\}, \forall n \in \mathcal{N}
\end{align}
where $\eta > 0$ is a constant with arbitrary value. Note that the active-sleep constraints in (\ref{eq:bigM onoff}) are inspired by the well-known big-$M$ method \cite{Integer}: if $\rho_n = 0$, the constraint (\ref{eq:bigM onoff}) ensures that $\mathbf{w}^{\text{DL}}_{i,n} = \mathbf{v}^{\text{UL}}_{i,n} = \mathbf{0}, \forall i \in \mathcal{K}$; if $\rho_n = 1$, the constraint has no effect on $\mathbf{w}^{\text{DL}}_{i,n}$ and $\mathbf{v}^{\text{UL}}_{i,n}, \forall i \in \mathcal{K}$, as $P^{\text{DL}}_{n,\text{max}}+\eta$ represents an upper bound on the term $\sum^{K}_{i=1}\|\mathbf{w}^{\text{DL}}_{i,n}\|^2+\|\mathbf{v}^{\text{UL}}_{i,n}\|^2$. Notice that $\eta$ can be chosen arbitrarily due to the scaling invariant property of UL receive beamforming vector $\mathbf{v}^{\text{UL}}_i$'s. With the active-sleep constraints in (\ref{eq:bigM onoff}), the equivalence between problems (P1) and (P2) can be easily verified.

In \cite{Cheng13}, a similar problem to (P2) was studied corresponding to the case with only DL transmission. For problem (P2) without $\mathbf{v}^{\text{UL}}_i$ and $p^{\text{UL}}_i, \forall i \in \mathcal{K}$ and their corresponding constraints, the problem can be transformed to a convex second-order cone programming (SOCP) by relaxing the binary variable $\rho_n$ as $\rho_n \in [0, 1], \forall n \in \mathcal{N}$. Under this convex relaxation, a BnC algorithm, which is a combination of the branch-and-bound (BnB) and the cutting plane (CP) methods \cite{Integer}, was proposed in \cite{Cheng13} to solve the DL problem optimally. However, the computational complexity of BnB is prohibitive for large networks in practice, which grows exponentially with the number of APs. To obtain polynomial-time algorithm with near-optimal performance, in \cite{Cheng13}, the authors further proposed an incentive measure based heuristic algorithm to determine the set of active APs. The incentive measure reflects the importance of each AP to the whole network and is defined as the ratio of the total power received at all MUs to the total power expended for each AP.

However, with both DL and UL transmissions, it is observed that problem (P2) can no longer be transformed to a convex form by relaxing $\rho_n$'s as continuous variables due to the constraints in (\ref{eq:ULDL UL SINR constraint}). Furthermore, because of the scaling invariant property of UL receive beamforming vectors, solving the relaxed problem of (P2) will result in all $\mathbf{v}^{\text{UL}}_i$'s going to zero, similar to the case of GSO based solution. Particularly, the value of the relaxed indicator $\rho_n, \forall n \in \mathcal{N}$, will not be related to $\mathbf{v}^{\text{UL}}_i$'s, which in fact contributes to the penalty incurred in the objective due to the static power of AP $n$, i.e., $\rho_nP_{c,n}$. Finally, it is nontrivial to find an incentive measure that reflects the importance of each AP to both DL and UL transmissions.

\section{Proposed Solution}\label{sec:joint UL DL}
In this section, we provide two efficient algorithms to approximately solve problem (P1) based on the GSO and RIP approaches, respectively.

\subsection{Proposed Algorithm for (P1) based on GSO}\label{sec:GSO}
First, we consider the approach of GSO and present a new method to address the joint DL and UL optimization. To obtain an efficient solution for problem (P1), we first assume that all the APs and MUs have infinite power budget, i.e., $P^{\text{DL}}_{n,\text{max}}=+\infty$, $\forall n \in \mathcal{N}$ and $P^{\text{UL}}_{i,\text{max}} = +\infty, \forall i \in \mathcal{K}$. The resulting problem is termed (P1-1). An equivalent reformulation of problem (P1-1) is then provided to overcome the receive beamforming scaling issue mentioned in Section \ref{sec:problem formulation}. Then, we discuss the challenges of dealing with finite per-AP and per-MU power constraints and provide efficient methods to handle them.

\subsubsection{Solution for problem (P1-1)}\label{sec:without}
First, we consider the following transmit sum-power minimization problem in the UL:

\begin{small}
\vspace{-0.2in}
\begin{align}\label{eq:UL power minimization}
\mathop{\mathtt{Min.}}\limits_{\{\mathbf{v}^{\text{UL}}_i\},\{p^{\text{UL}}_i\}} &
~~ \sum_{i=1}^{K}p^{\text{UL}}_i \nonumber \\
\mathtt{s.t.}
& ~~ \text{SINR}^{\text{UL}}_i \geq \gamma^{\text{UL}}_i, \forall i \in \mathcal{K} \nonumber \\
& ~~ p^{\text{UL}}_i \geq 0, \forall i \in \mathcal{K}.
\end{align}
\end{small}%
From \cite{Boche05}, it follows that Problem (\ref{eq:UL power minimization}) can be solved in a virtual DL channel as
\begin{small}
\begin{align}\label{eq:virtual DL power minimizatin}
\mathop{\mathtt{Min.}}\limits_{\{\mathbf{w}^{\text{VDL}}_i\}} &
~~ \sum_{i=1}^{K}\|\mathbf{w}^{\text{VDL}}_i\|^2 \nonumber \\
\mathtt{s.t.}
& ~~ \text{SINR}^{\text{VDL}}_i \triangleq \frac{|\mathbf{g}_i^{H}\mathbf{w}^{\text{VDL}}_i|^2}{\sum_{j \neq i}|\mathbf{g}_i^{H}\mathbf{w}^{\text{VDL}}_j|^2 + \sigma^2} \geq \gamma^{\text{UL}}_i, \forall i \in \mathcal{K}
\end{align}
\end{small}%
where $\mathbf{w}^{\text{VDL}}_i \in \mathbb{C}^{M\times1}$ is the virtual DL transmit beamforming vector over $N$ APs for MU $i$. Denote $(\mathbf{v}^{\text{UL}}_i)^{'}$, $(p^{\text{UL}}_i)^{'}$ and $(\mathbf{w}^{\text{VDL}}_i)^{'}$, $i = 1,\cdots,K$ as the optimal solutions to problems (\ref{eq:UL power minimization}) and (\ref{eq:virtual DL power minimizatin}), respectively. Then from \cite{Boche05} it follows that $(\mathbf{v}^{\text{UL}}_i)^{'}$ and $(\mathbf{w}^{\text{VDL}}_i)^{'}$ can be set to be identical, $i = 1,\cdots,N$, and furthermore $\sum_{i=1}^{K}(p^{\text{UL}}_i)^{'} = \sum_{i=1}^{K}\|(\mathbf{w}^{\text{VDL}}_i)^{'}\|^2$.

By establishing a virtual DL transmission for the UL transmission based on the above UL-DL duality, we have the following lemma.
\begin{lemma}\label{lemma:1}
Problem (P1-1) is equivalent to the following problem.
\begin{align}
(\text{P3}):~\mathop{\mathtt{Min.}}\limits_{\{\mathbf{w}^{\text{DL}}_i\},\{\mathbf{w}^{\text{VDL}}_i\}}&
~~ \sum^{N}_{n=1}\mathbf{1}_n\left(\{\mathbf{w}^{\text{DL}}_{i,n}\},\{\mathbf{w}^{\text{VDL}}_{i,n}\}\right)P_{c,n} \nonumber \\ &+\sum^{K}_{i=1}\|\mathbf{w}^{\text{DL}}_{i}\|^2+\lambda\sum_{i=1}^{K}\|\mathbf{w}^{\text{VDL}}_i\|^2\\
\mathtt{s.t.}
& ~~  \text{SINR}^{\text{DL}}_i \geq \gamma^{\text{DL}}_i, \forall i \in \mathcal{K} \label{eq:DL SINR r}\\
& ~~ \text{SINR}^{\text{VDL}}_i \geq \gamma^{\text{UL}}_i, \forall i \in \mathcal{K}. \label{eq:VDL SINR r}
\end{align}
\end{lemma}
\begin{proof}
{\color{red}For any given feasible solution to problem (P3), we can always find a corresponding feasible solution to problem (P1-1) achieving the same objective value as that of problem (P3), and vice versa, similar as \cite[Proposition 1]{RuiIt}; thus, problems (P1-1) and (P3) achieve the same optimal value with the same set of optimal DL/UL beamforming vectors. Lemma \ref{lemma:1} is thus proved.}
\end{proof}

Since problem (P3) is a DL-only problem that has the same ``group-sparse'' property as (P1-1), it can be approximately solved by replacing the objective function with
\begin{align}\label{eq:reformulate problem objective}
\sum^{N}_{n=1}\beta_n\sqrt{\sum^{K}_{i=1}\|\mathbf{w}^{\text{DL}}_{i,n}\|^2+\|\mathbf{w}^{\text{VDL}}_{i,n}\|^2} + \sum_{i=1}^{K}\|\mathbf{w}^{\text{DL}}_i\|^2 + \lambda\sum_{i=1}^{K}\|\mathbf{w}^{\text{VDL}}_i\|^2.
\end{align}
Comparing (\ref{eq:reformulate problem objective}) and (\ref{eq:scaling problem objective}), we have successfully solved the scaling issue of UL receive beamforming vector, $\mathbf{v}^{\text{UL}}_i$'s, by replacing them with the equivalent DL transmit beamforming vector, $\mathbf{w}^{\text{VDL}}_i$'s, since from (\ref{eq:virtual DL power minimizatin}) it follows that the virtual DL SINR of each MU $i$ is no more scaling invariant to $\mathbf{w}^{\text{VDL}}_i$'s.

Furthermore, since any arbitrary phase rotation of the beamforming vectors does not affect both (\ref{eq:reformulate problem objective}) and the SINR constrains in (\ref{eq:DL SINR r}) and (\ref{eq:VDL SINR r}), (P3) with (\ref{eq:reformulate problem objective}) as the objective function can be reformulated as a convex SOCP \cite{Boydbook}, which is given by

\begin{small}
\vspace{-0.1in}
\begin{align}
&(\text{P4}): \nonumber \\~&\mathop{\mathtt{Min.}}\limits_{\{\mathbf{w}^{\text{DL}}_i\},\{\mathbf{w}^{\text{VDL}}_i\},\{t_n\}}
~~ \sum^{N}_{n=1}\beta_nt_n + \sum_{i=1}^{K}\|\mathbf{w}^{\text{DL}}_i\|^2 + \lambda\sum_{i=1}^{K}\|\mathbf{w}^{\text{VDL}}_i\|^2\\
\mathtt{s.t.}
& ~~ \left\|\begin{array}{cl} \displaystyle
      \mathbf{h}_i^{H}\mathbf{W}^{\text{DL}} \\
        \sigma \end{array}\right\| \leq \sqrt{1+\frac{1}{\gamma^{\text{DL}}_i}}\mathbf{h}_i^{H}\mathbf{w}^{\text{DL}}_i, \forall i \in \mathcal{K} \label{eq:DL-SOCP SINR constraint}\\
& ~~ \left\|\begin{array}{cl} \displaystyle
      \mathbf{g}_i^{H}\mathbf{W}^{\text{VDL}} \\
        \sigma \end{array}\right\| \leq \sqrt{1+\frac{1}{\gamma^{\text{UL}}_i}}\mathbf{g}_i^{H}\mathbf{w}^{\text{VDL}}_i, \forall i \in \mathcal{K} \label{eq:VDL-SOCP SINR constraint}\\
& ~~ \sqrt{\sum^{K}_{i=1}\|\mathbf{w}^{\text{DL}}_{i,n}\|^2+\|\mathbf{w}^{\text{VDL}}_{i,n}\|^2} \leq t_n, \forall n \in \mathcal{N} \label{eq:slack}
\end{align}
\end{small}%
where $\mathbf{W}^{\text{DL}} = [\mathbf{w}^{\text{DL}}_1, \cdots, \mathbf{w}^{\text{DL}}_K]$, $\mathbf{W}^{\text{VDL}} = [\mathbf{w}^{\text{VDL}}_1, \cdots, \mathbf{w}^{\text{VDL}}_K]$, and $t_n$'s are auxiliary variables with $t_n=0$ and $t_n>0$ indicating that AP $n$ is in active and sleep mode, respectively. {\color{red}Notice that without $\ell_{1,2}$ norm penalty or $\beta_n = 0$, $\forall n \in \mathcal{N}$, problem (P4) can be decomposed into two separate minimum-power beamforming design problems: one for the original DL transmission, and the other for the virtual DL transmission.}
\begin{remark}
{\color{red}Conventionally, the UL transmit sum-power minimization problem, as in (\ref{eq:UL power minimization}), has a convenient analytical structure and thus is computationally easier to handle, as compared to the DL minimum-power beamforming design problem, as in (\ref{eq:virtual DL power minimizatin}). Consequently, most existing studies in the literature have transformed the DL problem to its virtual UL formulation for convenience. The motivation of exploiting the reverse direction in this work, however, is to overcome the scaling issue of UL receive beamforming in GSO, so that we can solve the AP selection problem jointly for both DL and UL transmissions.}
\end{remark}

Next, we present the complete algorithm for problem (P1-1) based on GSO, in which three steps need to be performed sequentially.
\begin{enumerate}
\item \emph{Identify the subset of active APs denoted as $\mathcal{N}_{\text{on}}$}. This can be done by iteratively solving problem (P4) with different $\beta_n$'s. Notice that how to set the parameter $\beta_n$'s in (P4) plays a key role in the resulting APs selection. To optimally set the values of $\beta_n$'s, we adopt an iterative method similar as in \cite{Boyd}, shown as follows. In the $l$th iteration, $l \geq 1$, $t_n^{(l)}$'s are obtained by solving Problem (P4) with $\beta_n = \beta_n^{(l)}, \forall n \in \mathcal{N}$. The $\beta_n^{(l)}$'s are derived from the solution $t_n^{(l-1)}$'s of the $(l-1)$th iteration as
    \begin{align}\label{eq:iterative update}
    \beta^{(l)}_n = \frac{P_{c,n}}{t^{(l-1)}_n+\varepsilon}, n = 1,\cdots,N
    \end{align}
    where $\varepsilon$ is a small positive number to ensure stability. Notice that the initial values of $t^{(0)}_n$'s are chosen as
    \begin{align}
    t^{(0)}_n = \sqrt{\sum^{K}_{i=1}\|\tilde{\mathbf{w}}^{\text{DL}}_{i,n}\|^2+\|\tilde{\mathbf{w}}^{\text{VDL}}_{i,n}\|^2}, n = 1,\cdots,N
    \end{align}
    where $\tilde{\mathbf{w}}^{\text{DL}}_{i,n}$ and $\tilde{\mathbf{w}}^{\text{VDL}}_{i,n}$ are the beamforming vector solution of Problem (P4) with $\beta_n = 0, \forall n \in \mathcal{N}$. The above update is repeated until $|\beta^{(l)}_n - \beta^{(l-1)}_n| < \eta$, $\forall n \in \mathcal{N}$, where $\eta$ is a small positive constant that controls the algorithm accuracy.

    Let $\mathbf{t}^{\star} = [t_1^{\star},\cdots,t_N^{\star}]$ denote the sparse solution after the convergence of the above iterative algorithm. Then the nonzero entries in $\mathbf{t}^{\star}$ correspond to the APs that need to be active, i.e., $\mathcal{N}_{\text{on}}= \left\{n|t^{\star}_n>0, n \in \mathcal{N}\right\}$.
\item \emph{Obtain the optimal transmit/receive beamforming vectors $\left(\mathbf{w}^{\text{DL}}_i\right)^{\star}$ and $\left(\mathbf{w}^{\text{VDL}}_i\right)^{\star}$, $i = 1,\cdots,K$, given the selected active APs}. This can be done by solving (P4) with $\beta_n = 0$, $\forall n \in \mathcal{N_{\text{on}}}$ and $ \mathbf{w}^{\text{DL}}_{i,n} = \mathbf{w}^{\text{VDL}}_{i,n} = \mathbf{0}, i=1,...,K, \forall n \notin \mathcal{N}_{\text{on}}$.
\item \emph{Obtain the optimal transmit power values of MUs $\left(p^{\text{UL}}_i\right)^{\star}$, $i = 1,\cdots,K$}. This can be done by solving problem (\ref{eq:UL power minimization}) with $\mathbf{v}^{\text{UL}}_i = \left(\mathbf{w}^{\text{VDL}}_i\right)^{\star}$, $\forall i \in \mathcal{K}$, which is a simple linear programming (LP) problem.
\end{enumerate}
The iterative update given in (\ref{eq:iterative update}) is designed to make small entries in $\left\{t_n\right\}^{N}_{n=1}$ converge to zero. Furthermore, as the updating evolves, the penalty associated with AP $n$ in the objective function, i.e., $\beta_nt_n$, will converge to two possible values:
\begin{align}
\beta_nt_n & \rightarrow \left\{ \begin{array}{cl} \displaystyle
P_{c,n} & \mbox{if } t^{\star}_n > 0, \mbox{i.e., AP} ~ n ~\mbox{is active} \\
0 & \mbox{otherwise. }
\end{array}\right.
\end{align}
In other words, only the active APs will incur penalties being the exact same values as their static power consumption, which has the same effect as the indicator function in problem (P1-1) or (P1). {\color{red}Convergence of this algorithm can be shown by identifying the iterative update as a Majorization-Minimization (MM) algorithm \cite{MM} for a concave minimization problem, i.e., using $\log(\cdot)$ function, which is concave, to approximate the indicator function given in (\ref{eq:on off condition}). The details are thus omitted due to space limitations.}

\subsubsection{Per-AP and Per-MU Power Constraints}\label{sec:with}
It is first observed that the per-AP power constraints in (\ref{eq:ULDL DL power constraint}), i.e.,
\begin{align}
\sum^{K}_{i=1}\|\mathbf{w}^{\text{DL}}_{i,n}\|^2 \leq P^{\text{DL}}_{n,\text{max}}, ~~ n = 1,\cdots,N
\end{align}
are convex. Therefore, adding per-AP power constraints to problem (P1-1) does not need to alter the above algorithm. Thus, we focus on the per-MU power constraints in the UL transmission in this subsection.

Again, we consider the following transmit sum-power minimization problem in the UL with per-MU power constraints:
\begin{align}\label{eq:UL power minimization with per MU}
\mathop{\mathtt{Min.}}\limits_{\{\mathbf{v}^{\text{UL}}_i\},\{p^{\text{UL}}_i\}} &
~~ \sum_{i=1}^{K}p^{\text{UL}}_i \nonumber \\
\mathtt{s.t.}
& ~~ \text{SINR}^{\text{UL}}_i \geq \gamma^{\text{UL}}_i, \forall i \in \mathcal{K} \nonumber \\
& ~~ 0 \leq p^{\text{UL}}_i \leq P^{\text{UL}}_{i,\text{max}}, \forall i \in \mathcal{K}.
\end{align}
{\color{red}Although it has been shown in \cite{Yu07} that the sum-power minimization problem in the DL with per-AP power constraints can be transformed into an equivalent min-max optimization problem in the UL, we are not able to find an equivalent DL problem for problem (\ref{eq:UL power minimization with per MU}) as in Section \ref{sec:without} which is able to handle the per-MU power constraints. The fundamental reason is that the power allocation obtained by solving problem (\ref{eq:UL power minimization}) is already component-wise minimum, which can be shown by the uniqueness of the fixed-point solution for a set of minimum SINR requirements in the UL given randomly generated channels \cite{Wiesel06}. The component-wise minimum power allocation indicates that it is not possible to further reduce one particular MU's power consumption by increasing others', i.e., there is no tradeoff among different MUs in terms of power minimization. Consequently, solving problem (\ref{eq:UL power minimization with per MU}) requires only one additional step compared with solving problem (\ref{eq:UL power minimization}), i.e., checking whether the optimal power solution to problem (\ref{eq:UL power minimization}) satisfies the per-MU power constraints. If this is the case, the solution is also optimal for problem (\ref{eq:UL power minimization with per MU}); otherwise, problem (\ref{eq:UL power minimization with per MU}) is infeasible.}

Next, we present our complete algorithm for problem (P1) with the per-AP and per-MU power constraints. Compared to the algorithm proposed for problem (P1-1) in Section \ref{sec:without} without the per-AP and per-MU power constraints, the new algorithm differs in the first step, i.e., to identify the subset of active APs. The main idea is that a set of candidate active APs is first obtained by ignoring the per-MU power constraints but with a new sum-power constraint in the UL (or equivalently its virtual DL), i.e., we iteratively solve the following problem similarly as in the first step of solving problem (P1-1) in Section \ref{sec:without}.
\begin{align}
(\text{P5}):\mathop{\mathtt{Min.}}\limits_{\{\mathbf{w}^{\text{DL}}_i\},\{\mathbf{w}^{\text{VDL}}_i\},\{t_n\}} &
~~ \sum^{N}_{n=1}\beta_nt_n + \sum_{i=1}^{K}\|\mathbf{w}^{\text{DL}}_i\|^2 + \lambda\sum_{i=1}^{K}\|\mathbf{w}^{\text{VDL}}_i\|^2\\
\mathtt{s.t.}
& ~~ (\ref{eq:DL-SOCP SINR constraint}), (\ref{eq:VDL-SOCP SINR constraint}) ~ \mbox{and} ~(\ref{eq:slack}) \\
& ~~ \sum^{K}_{i=1}\|\mathbf{w}^{\text{DL}}_{i,n}\|^2 \leq P^{\text{DL}}_{n,\text{max}}, \forall n \in \mathcal{N} \label{eq:per AP p4}\\
& ~~ \sum^{K}_{i=1}\|\mathbf{w}^{\text{VDL}}_i\|^2 \leq \sum^{K}_{i=1}P^{\text{UL}}_{i,\text{max}} \label{eq:sum user p4}.
\end{align}
The sum-power constraint in (\ref{eq:sum user p4}) is added to impose a mild control on the transmit powers of all MUs in the UL. After obtaining the candidate set, the feasibility of the UL transmission is then verified. If the candidate set can support the UL transmission with the given per-MU power constraints, then the optimal solution of (P1) is obtained; otherwise, one or more APs need to be active for the UL transmission.

To be more specific, denote the set of candidate active APs obtained by iteratively solving problem (P5) as $\mathcal{\tilde{N}}_{\text{on}}$. Problem (\ref{eq:UL power minimization}) is then solved with $\mathcal{\tilde{N}}_{\text{on}}$, for which the feasibility is guaranteed due to the virtual DL SINR constraints in (\ref{eq:VDL-SOCP SINR constraint}). We denote the obtained power allocation as $\tilde{p}^{\text{UL}}_i$, $i = 1,\cdots, K$.
\begin{itemize}
\item If $\mathcal{\tilde{N}}_{\text{on}}$ can support the UL transmission without violating any MU's power constraint, i.e.,
\begin{align}
\tilde{p}^{\text{UL}}_i \leq P^{\text{UL}}_{i,\text{max}}, \forall i \in \mathcal{K}
\end{align}
the candidate set can be finalized as the set of active APs and the algorithm proceeds to find the optimal transmit/receive beamforming vectors similarly as that in Section \ref{sec:without}.

\item If $\mathcal{\tilde{N}}_{\text{on}}$ cannot support the UL transmission with the given MU's power constraints, we propose the following price based iterative method to determine the additional active APs. Specifically, in each iteration, for those APs that are not in the candidate set each will be assigned a price $\theta_m, m \notin \mathcal{\tilde{N}}_{\text{on}}$, which is defined as
\begin{align}\label{eq:price}
\theta_m = \frac{1}{P_{c,m}}\sum_{i \in \mathcal{B}}\frac{\tilde{p}^{\text{UL}}_i-P^{\text{UL}}_{i,\text{max}}}{P^{\text{UL}}_{i,\text{max}}}\|\mathbf{g}_{i,m}\|^2, ~ \forall m \notin \mathcal{\tilde{N}}_{\text{on}}
\end{align}
where $\mathcal{B} \triangleq \left\{i | \tilde{p}^{\text{UL}}_i > P^{\text{UL}}_{i,\text{max}}, i \in \mathcal{K}\right\}$. The price $\theta_m$ is set to be the normalized (by its corresponding static power consumption) weighted-sum power gains of the channels from AP $m$ to all the MUs that have their power constraints being violated. The weights are chosen as the ratios of MUs' required additional powers to their individual power limits. According to the definition of $\theta_m$ in (\ref{eq:price}), the AP having smaller static power consumption and better channels to MUs whose power constraints are more severely violated will be associated with a larger price. The candidate set is then updated by including the AP that corresponds the largest $\theta_m$ as
\begin{align}
\mathcal{\tilde{N}}_{\text{on}} \leftarrow \mathcal{\tilde{N}}_{\text{on}} \cup \left(\arg\max\limits_{m\notin \mathcal{\tilde{N}}_{\text{on}}}{\theta_m}\right).
\end{align}
With updated $\mathcal{\tilde{N}}_{\text{on}}$, the feasibility of the UL transmission needs to be re-checked by obtaining a new set of power allocation, which will be used to compute the new $\theta_m$'s in next iteration if further updating is required. The above process is repeated until all the MUs' power constraints are satisfied. Its convergence is guaranteed since problem (P1) has been assumed to be feasible if all APs are active.
\end{itemize}
Combining with the algorithm in Section \ref{sec:without}, our complete algorithm for problem (P1) based on GSO is summarized in Table \ref{table1}. {\color{red}For the algorithm given in Table \ref{table1}, there are two problems that need to be iteratively solved, i.e., problems (\ref{eq:UL power minimization}) and (P5). Since problem (\ref{eq:UL power minimization}) can be efficiently solved by the fixed-point algorithm \cite{Wiesel06}, the computation time is dominated by solving the SOCP problem (P5). If the primal-dual interior point algorithm \cite{Boydbook} is used by the numerical solver for solving (P5), the computational complexity is of order $M^{3.5}K^{3.5}$. Furthermore, since the convergence of the iterative update in steps 4)-5), governed by the MM algorithm, is very fast (approximately $10$-$15$ iterations) as observed in the simulations, the overall complexity of the algorithm in Table \ref{table1} is approximately $\mathcal{O}(M^{3.5}K^{3.5})$.}

\begin{table}[H]
\begin{center}
\caption{\textbf{Algorithm \ref{table1}}: Proposed algorithm for Problem (P1) based on GSO} \vspace{0.2cm}
 \hrule
\vspace{0.3cm}
\begin{enumerate}
\item Set $l=0$, initialize the set of candidate active APs as $\mathcal{\tilde{N}}_{\text{on}} = \mathcal{N}$.
\item Obtain $\tilde{\mathbf{w}}^{\text{DL}}_{i,n}$'s and $\tilde{\mathbf{w}}^{\text{VDL}}_{i,n}$'s by solving problem (P5) with $\beta_n = 0, \forall n \in \mathcal{N}$.
\item Set $t^{(0)}_n = \sqrt{\sum^{K}_{i=1}\|\tilde{\mathbf{w}}^{\text{DL}}_{i,n}\|^2+\|\tilde{\mathbf{w}}^{\text{VDL}}_{i,n}\|^2}, n = 1,\cdots,N$.
\item {\bf Repeat:}
    \begin{itemize}
    \item[ a)] $l \leftarrow l+1$.
    \item[ b)] Set $\beta^{(l)}_n = \frac{P_{c,n}}{t^{(l-1)}_n+\varepsilon}$, $\forall n \in \mathcal{N}$.
    \item[ c)] Obtain $\mathbf{t}^{(l)} = [t^{(l)}_1,\cdots,t^{(l)}_N]$ by solving problem (P5) with $\beta_n = \beta^{(l)}_n, \forall n \in \mathcal{N}$.
    \end{itemize}
\item {\bf Until} $|\beta^{(l)}_n - \beta^{(l-1)}_n| \leq \eta$, $\forall n \in \mathcal{N}$ or $l = l_{\text{max}}$.
\item Set $\mathcal{\tilde{N}}_{\text{on}}$ as $\mathcal{\tilde{N}}_{\text{on}}= \left\{n|t^{\star}_n>0, n \in \mathcal{N}\right\}$.
\item {\bf Repeat:}
    \begin{itemize}
        \item[ a)] Obtain $\tilde{p}^{\text{UL}}_i$, $i = 1,\cdots,K$, by solving problem (\ref{eq:UL power minimization}) with $\mathcal{\tilde{N}}_{\text{on}}$.
        \item[ b)] Set $\mathcal{B} = \left\{i | \tilde{p}^{\text{UL}}_i > P^{\text{UL}}_{i,\text{max}}, i \in \mathcal{K}\right\}$.
        \item[ c)] Set $\theta_m = \frac{1}{P_{c,m}}\sum_{i \in \mathcal{B}}\frac{\tilde{p}^{\text{UL}}_i-P^{\text{UL}}_{i,\text{max}}}{P^{\text{UL}}_{i,\text{max}}}\|\mathbf{g}_{i,m}\|^2, \forall m \notin \mathcal{\tilde{N}}_{\text{on}}$.
        \item[ d)] Set $\mathcal{\tilde{N}}_{\text{on}} \leftarrow \mathcal{\tilde{N}}_{\text{on}} \cup \left(\arg\max\limits_{m\notin \mathcal{\tilde{N}}_{\text{on}}}{\theta_m}\right)$.
    \end{itemize}
\item {\bf Until} $\mathcal{B} = \emptyset$.
\item Obtain $\left(\mathbf{w}^{\text{DL}}_i\right)^{\star}$ and $\left(\mathbf{w}^{\text{V-DL}}_i\right)^{\star}$, $i = 1,\cdots,K$, by solving (P5) with $\beta_n = 0$, $\forall n \in \mathcal{N}$ and $ \mathbf{w}^{\text{DL}}_{i,n} = \mathbf{w}^{\text{VDL}}_{i,n} = \mathbf{0}, i=1,...,K, \forall n \notin \mathcal{\tilde{N}}_{\text{on}}$.
\item Set $\left(\mathbf{v}^{\text{UL}}_i\right)^{\star} = \left(\mathbf{w}^{\text{VDL}}_i\right)^{\star}, \forall i \in \mathcal{K}$, and compute the $\left(p^{\text{UL}}_i\right)^{\star}$, $i = 1,\cdots,K$, by solving problem (\ref{eq:UL power minimization}).
\end{enumerate}
\vspace{0.2cm} \hrule \label{table1} \end{center}
\end{table}

\subsection{Proposed Algorithm for (P1) based on RIP}
In this subsection, an alternative algorithm for problem (P1) is developed based on RIP by applying the same idea of establishing a virtual DL transmission for the original UL. Similar to the case with GSO, the per-MU power constraints are first replaced with a sum-power constraint in the UL. The resulting problem is further reformulated as a convex SOCP by relaxing the binary variables $\{\rho_{n}\}$, which is given as follows.
\begin{align}
(\text{P6}):\mathop{\mathtt{Min.}}\limits_{\{\mathbf{w}^{\text{DL}}_i\},\{\mathbf{w}^{\text{VDL}}_i\},\{\rho_n\}} &
~~  \sum^{N}_{n=1}\rho_nP_{c,n} + \sum_{i=1}^{K}\|\mathbf{w}^{\text{DL}}_i\|^2 + \lambda\sum_{i=1}^{K}\|\mathbf{w}^{\text{VDL}}_i\|^2 \nonumber \\
\mathtt{s.t.}
& ~~ (\ref{eq:DL-SOCP SINR constraint}), (\ref{eq:VDL-SOCP SINR constraint}) ~ \text{and} ~ (\ref{eq:sum user p4}) \nonumber \\
& ~~ \sum^{K}_{i=1}\|\mathbf{w}^{\text{DL}}_{i,n}\|^2 \leq \rho_nP^{\text{DL}}_{n,\text{max}}, \forall n \in \mathcal{N} \label{eq:bigm1}\\
& ~~ \sum^{K}_{i=1}\|\mathbf{w}^{\text{VDL}}_{i,n}\|^2 \leq \rho_n\sum^{K}_{i=1}P^{\text{UL}}_{i,\text{max}}, \forall n \in \mathcal{N} \label{eq:bigm2}\\
& ~~ \sum^{N}_{n = 1}\rho_n \geq 1 \label{eq:red}\\
& ~~ 0 \leq \rho_n \leq 1, \forall n \in \mathcal{N}. \label{eq:bigml}
\end{align}
Note that instead of implementing the active-sleep constraints jointly for the actual DL and virtual DL as (\ref{eq:bigM onoff}) in problem (P2), i.e., $\sum^{K}_{i=1}\|\mathbf{w}^{\text{DL}}_{i,n}\|^2+\|\mathbf{w}^{\text{VDL}}_{i,n}\|^2 \leq \rho_n(P^{\text{DL}}_{n,\text{max}}+\sum^{K}_{i=1}P^{\text{UL}}_{i,\text{max}})$, we divide them into two sets of coupled active-sleep constraints as in (\ref{eq:bigm1}) and (\ref{eq:bigm2}) via $\rho_n$'s. For the non-relaxed problem of (P6) with binary $\rho_n$'s, i.e. $\rho_n \in \{0, 1\}, \forall n \in \mathcal{N}$, it can be shown that these two formulations are equivalent. However, for the case of the relaxed problem (P6) with continuous valued $\rho_n$'s, the separated active-sleep constraints are designed to avoid the situation that the difference between $P^{\text{DL}}_{n,\text{max}}$ and $\sum^{K}_{i=1}P^{\text{UL}}_{i,\text{max}}$ is too large such that the optimal value of $\rho_n$ is dominated by either DL or UL transmission. To implement the big-$M$ method with active-sleep constraints \cite{Integer}, an appropriate upper bound for the term $\sum^{K}_{i=1}\|\mathbf{w}^{\text{VDL}}_{i,n}\|^2$ needs to be found. According to the UL-DL duality, the minimum sum-power achieved is the same for the UL and its virtual DL transmissions. Therefore, $\sum^{K}_{i=1}P^{\text{UL}}_{i,\text{max}}$ can be chosen as the upper bound of $\sum^{K}_{i=1}\|\mathbf{w}^{\text{VDL}}_{i,n}\|^2, \forall n \in \mathcal{N}$. Finally, it is evident that the optimal value of problem (P6) serves as a lower bound of its non-relaxed problem with binary $\rho_n$'s. In order to further tighten this lower bound, one way is to reduce the feasible set of design variables. Constraint in (\ref{eq:red}) is introduced specifically to achieve this end, which can be shown to be redundant for the non-relaxed problem of (P6).

We adopt the same idea of incentive measure based AP selection as in \cite{Cheng13} to design a polynomial-time algorithm for problem (P2). However, it remains to find an incentive measure that reflects the importance of each AP to both DL and UL transmissions based on problem (P6). It is interesting to observe that after transforming the UL related terms to their virtual DL counterparts, the optimal relaxed binary variable solution of problem (P6) becomes a good choice to serve this purpose. Let $\breve{\rho}_n$, $\breve{\mathbf{w}}^{\text{DL}}_i$ and $\breve{\mathbf{w}}^{\text{VDL}}_i$ denote the optimal solution to problem (P6). Intuitively, the AP that has larger static power consumption and worse channels to all MUs is more desired to be switched into sleep mode from the perspective of energy saving. In (P6), for AP $n$ having larger $P_{c,n}$, $\breve{\rho}_n$ is desired to be smaller in order to achieve the minimum value of the objective function. Furthermore, it is practically valid that for DL power minimization problem, the optimal transmit power of APs that have worse channels to MUs is in general smaller. As a result, solving problem (P6) yields smaller $\sum^{K}_{i=1}\|\breve{\mathbf{w}}^{\text{DL}}_{i,n}\|^2$ and $\sum^{K}_{i=1}\|\breve{\mathbf{w}}^{\text{VDL}}_{i,n}\|^2$ and thus smaller $\breve{\rho}_n$ for AP $n$ that has worse channels to MUs for both the DL and UL transmissions. To summarize, the AP that corresponds to smaller $\breve{\rho}_n$ is more desired to be switched into sleep mode.

An iterative process is then designed to determine the set of active APs based on $\breve{\rho}_n$'s that are taken as incentive measures. The process starts with assuming all APs are active. In each iteration, problem (P6) is solved with a candidate set of active APs, and the AP corresponding to the smallest $\breve{\rho}_n$ will be removed from the candidate set. This process is repeated until one of the following conditions occurs:
\begin{itemize}
\item The weighted sum-power cannot be further reduced;
\item Problem (P6) becomes infeasible;
\item Problem (\ref{eq:UL power minimization with per MU}) becomes infeasible.
\end{itemize}
Note that the feasibility checking for problem (\ref{eq:UL power minimization with per MU}) is the same as that in Algorithm I, which ensures the per-MU power constraints. An overall algorithm for problem (P1) based on RIP is summarized in Table \ref{table2}. {\color{red}For the algorithm given in Table \ref{table2}, the computation time is dominated by solving the SOCP problem (P6). If the primal-dual interior point algorithm \cite{Boydbook} is used by the numerical solver for solving (P6), the computational complexity is of order $M^{3.5}K^{3.5}$. Furthermore, since the worst case complexity for the iteration in steps 2)-3) is $\mathcal{O}(N)$, the overall complexity of the algorithm in Table \ref{table2} is $\mathcal{O}(NM^{3.5}K^{3.5})$.}

\begin{table}[H]
\begin{center}
\caption{\textbf{Algorithm \ref{table2}}: Proposed algorithm for Problem (P1) based on RIP} \vspace{0.2cm}
 \hrule
\vspace{0.3cm}
\begin{enumerate}
\item Set $l=0$, $\Phi^{(0)}$ a sufficiently large value, and initialize the set of candidate active APs as $\mathcal{\tilde{N}}_{\text{on}} = \mathcal{N}$.
\item {\bf Repeat:}
    \begin{itemize}
    \item[ a)] $l \leftarrow l+1$.
    \item[ b)] Solve problem (P6) with $\rho_n = 0$, $\forall n \notin \mathcal{\tilde{N}}_{\text{on}}$.
    \item[ c)] Set $\mathcal{\tilde{N}}_{\text{on}} \leftarrow \mathcal{\tilde{N}}_{\text{on}} \setminus \left(\arg\min\limits_{n\in \mathcal{\tilde{N}}_{\text{on}}}{\breve{\rho}^{(l)}_n}\right)$.
    \item[ d)] Set $\Phi^{(l)}$ as the optimal value of problem (P6) with $\rho_n = 1, \forall n \in \mathcal{\tilde{N}}_{\text{on}}$ and $\rho_n = 0, \forall n \notin \mathcal{\tilde{N}}_{\text{on}}$.
    \end{itemize}
\item {\bf Until} $\Phi^{(l)} > \Phi^{(l-1)}$ or problem (P6) is infeasible or problem (\ref{eq:UL power minimization with per MU}) is infeasible.
\item Obtain $\left(\mathbf{w}^{\text{DL}}_i\right)^{\star}$ and $\left(\mathbf{w}^{\text{V-DL}}_i\right)^{\star}$, $i = 1,\cdots,K$, by solving problem (P6) with $\rho_n = 1, \forall n \in \mathcal{\tilde{N}}_{\text{on}}$ and $\rho_n = 0, \forall n \notin \mathcal{\tilde{N}}_{\text{on}}$.
\item Set $\left(\mathbf{v}^{\text{UL}}_i\right)^{\star} = \left(\mathbf{w}^{\text{VDL}}_i\right)^{\star}, \forall i \in \mathcal{K}$, and compute the $\left(p^{\text{UL}}_i\right)^{\star}$, $i = 1,\cdots,K$, by solving problem (\ref{eq:UL power minimization}).
\end{enumerate}
\vspace{0.2cm} \hrule \label{table2} \end{center}
\end{table}

\section{Numerical Results}\label{sec:numerical}
In this section, we present numerical results to verify our proposed algorithms from three perspectives: ensuring feasibility for both DL and UL transmissions; achieving network power saving with optimal MU-AP association; and adjusting minimum power consumption tradeoffs between active APs and MUs. We consider two possible C-RAN configurations:
\begin{enumerate}
\item Homogeneous setup: all APs are assumed to have the same power consumption model with $P_{c,n} = 2$W and $P^{\text{DL}}_{n,\text{max}} = 1$W, $\forall n \in \mathcal{K}$, if not specified otherwise.

\item Heterogeneous setup: two types of APs are assumed, namely, high-power AP (HAP) and low-power AP (LAP), where the static power consumption for HAP and LAP are set as $50$W and $2$W, respectively, and the transmit power budgets for HAP and LAP are set as $20$W and $1$W, respectively.
\end{enumerate}
We assume that each AP $n$, $n \in \mathcal{N}$, is equipped with $M_{n} = 2$ antennas. For the single-antenna MU, we set the transmit power limit as $P^{\text{UL}}_{i,\text{max}} = 0.5$W, $\forall i \in \mathcal{K}$. For simplicity, we assume that the SINR requirements of all MUs are the same in the UL or DL. All the APs (except HAPs under the heterogeneous setup) and MUs are assumed to be uniformly and independently distributed in a square area with the size of $3$Km$\times$$3$Km. For all the simulations under heterogeneous setup, it is assumed that there are $2$ HAPs with fixed location at $[-750\text{m}, 0\text{m}]$ and $[750\text{m}, 0\text{m}]$, respectively. We assume a simplified channel model consisting of the distance-dependent attenuation with pathloss exponent $\alpha = 3$ and a multiplicative random factor (exponentially distributed with unit mean) accounting for short-term Rayleigh fading. We also set $\lambda = 1$ if not specified otherwise, i.e., we consider the sum-power consumption of all active APs and MUs. Finally, we set the receiver noise power for all the APs and MUs as $\sigma^2 = -50$dBm.

\subsection{Feasibility Performance}
First, we demonstrate the importance of active AP selection by jointly considering both DL and UL transmission in terms of the SINR feasibility in C-RAN. Since feasibility is our focus here instead of power consumption, it is assumed that the selected active APs will support all MUs for both the DL and UL transmissions. The simulation results compare our proposed algorithms (i.e., Algorithms I and II) with the following three AP selection schemes:
\begin{itemize}
\item {\bf AP initiated reference signal strength (APIRSS) based selection}: In this scheme, APs first broadcast orthogonal reference signals. Then, for each MU, the AP corresponding to the largest received reference signal strength will be included in the set of active APs. Note that this scheme has been implemented in practical cellular systems \cite{cellassociation}.
\item {\bf MU initiated reference signal strength (MUIRSS) based selection}: In this scheme, MUs first broadcast orthogonal reference signals. Then, for each MU, the AP corresponding to the largest received reference signal strength will be included in the set of active APs. Note that since all MUs are assumed to transmit reference signals with equal power and pathloss in general dominates short-term fading, the AP that is closest to each MU will receive strongest reference signal in general. Also note that in the previous APIRSS based scheme, if all APs are assumed to transmit with equal reference signal power (e.g., for the homogenous setup), the selected active APs will be very likely to be the same as those by the MUIRSS based scheme.
\item {\bf Proposed algorithm without considering UL (PAw/oUL)}: In this algorithm, the set of active APs are chosen from the conventional DL perspective by modifying our proposed algorithms. Specifically, Algorithm I is used here and similar results can be obtained with Algorithm II. Note that Algorithm I without considering UL transmission is similar to that proposed in \cite{Letaief13}.
\end{itemize}
With the obtained set of active APs, the feasibility check of problem (P1) can be decoupled into two independent feasibility problems: one for the DL and the other for the UL, while the network feasibility is achieved only when both the UL and DL SINR feasibility of all MUs are guaranteed.

In Fig. \ref{fig:FeasibilityDemo}, we illustrate the set of active APs generated by different schemes under the heterogeneous setup, and also compare them with that by the optimal exhaustive search. It is assumed that there are $2$ HAPs and $8$ LAPs jointly supporting $8$ MUs. The SINR targets for both DL and UL transmissions of all MUs are set as $8$dB. First, it is observed that Algorithm I and Algorithm II obtain the same set of active APs as shown in Fig. \ref{fig:FeasibilityDemoHetNetL12}, which is also identical to that found by exhaustive search. Second, it is observed that the $2$ HAPs are both chosen to be active in Fig. \ref{fig:FeasibilityDemoHetNetHighSignal} for the APIRSS based scheme. This is due to the significant difference between HAP and LAP in terms of transmit power, which makes most MUs receive the strongest DL reference signal from the HAP. The above phenomenon is commonly found in heterogenous network (HetNet) \cite{DenseHetnet} with different types of BSs (e.g. macro/micro/pico BSs). Third, from Fig. \ref{fig:FeasibilityDemoHetNetClose}, the active APs by the MUIRSS based scheme are simply those closer to the MUs, which is as expected. Finally, in Fig. \ref{fig:FeasibilityDemoHetNetDL}, only two LAPs are chosen to support all MUs with the PAw/oUL algorithm. This is because the algorithm does not consider UL transmission, and as a result Fig. \ref{fig:FeasibilityDemoHetNetDL} only shows the most energy-efficient AP selection for DL transmission.

\begin{figure}
\centering
\subfigure[]{
\includegraphics[width=0.48\textwidth]{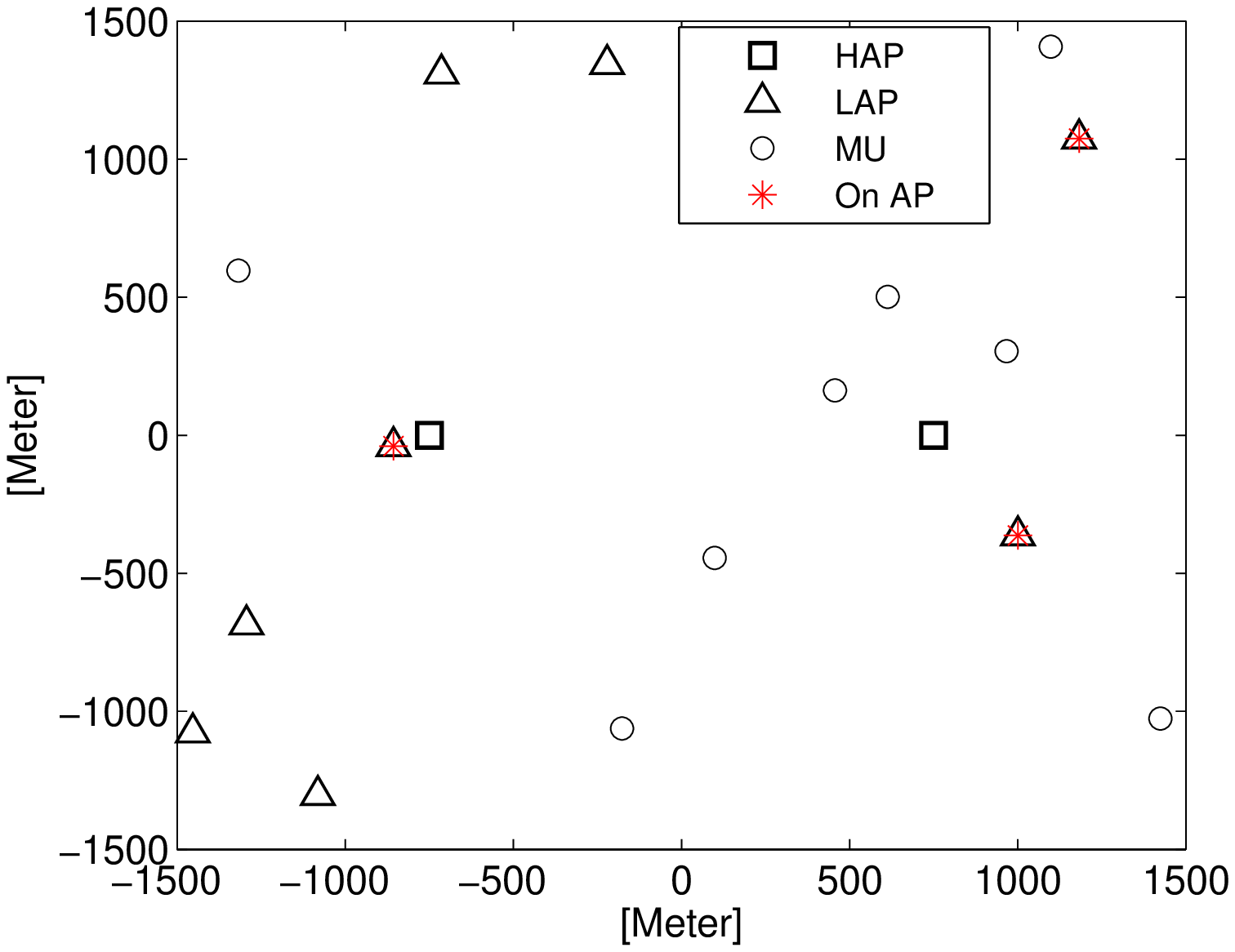}
\label{fig:FeasibilityDemoHetNetL12}
}
\subfigure[]{
\includegraphics[width=0.48\textwidth]{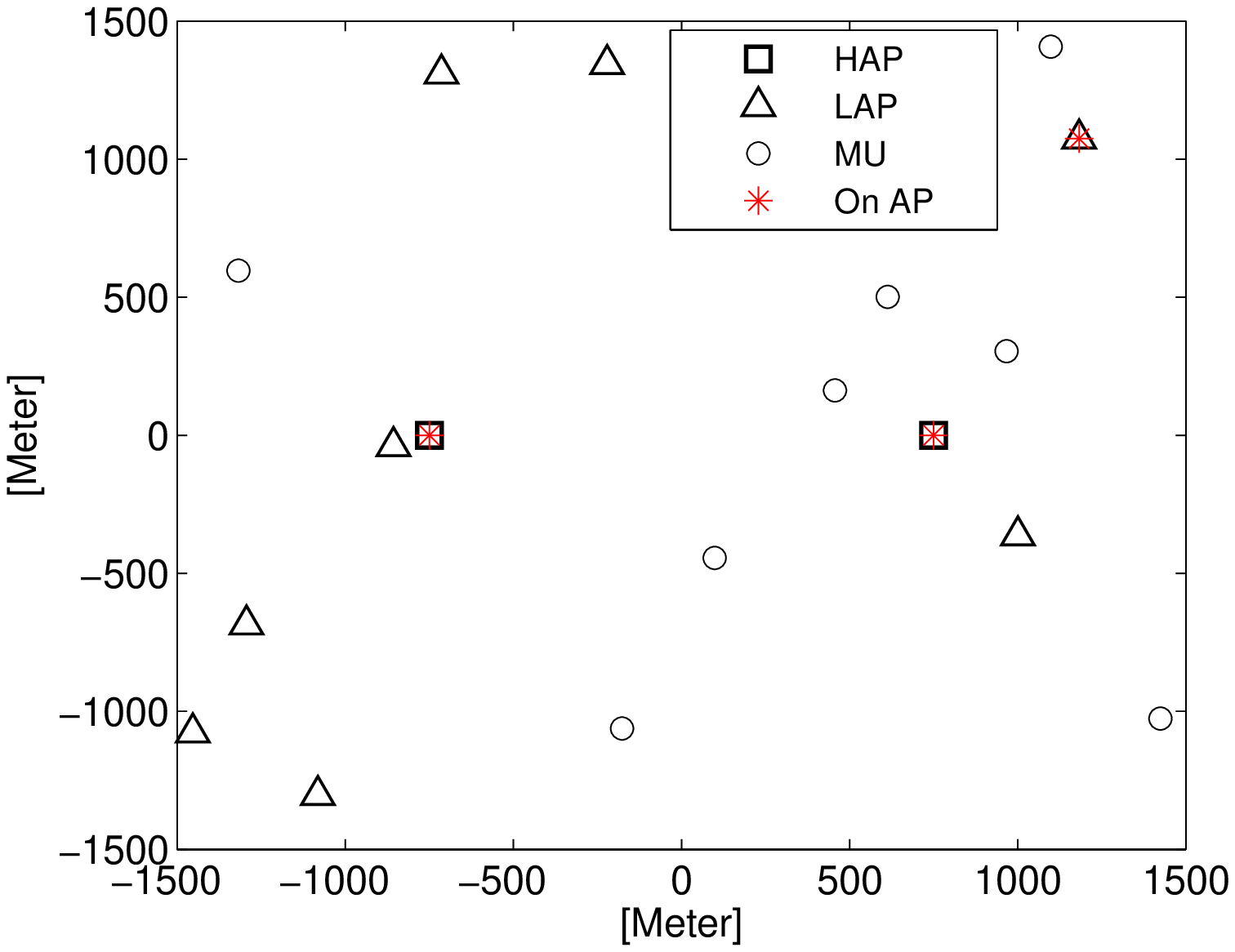}
\label{fig:FeasibilityDemoHetNetHighSignal}
}
\subfigure[]{
\includegraphics[width=0.48\textwidth]{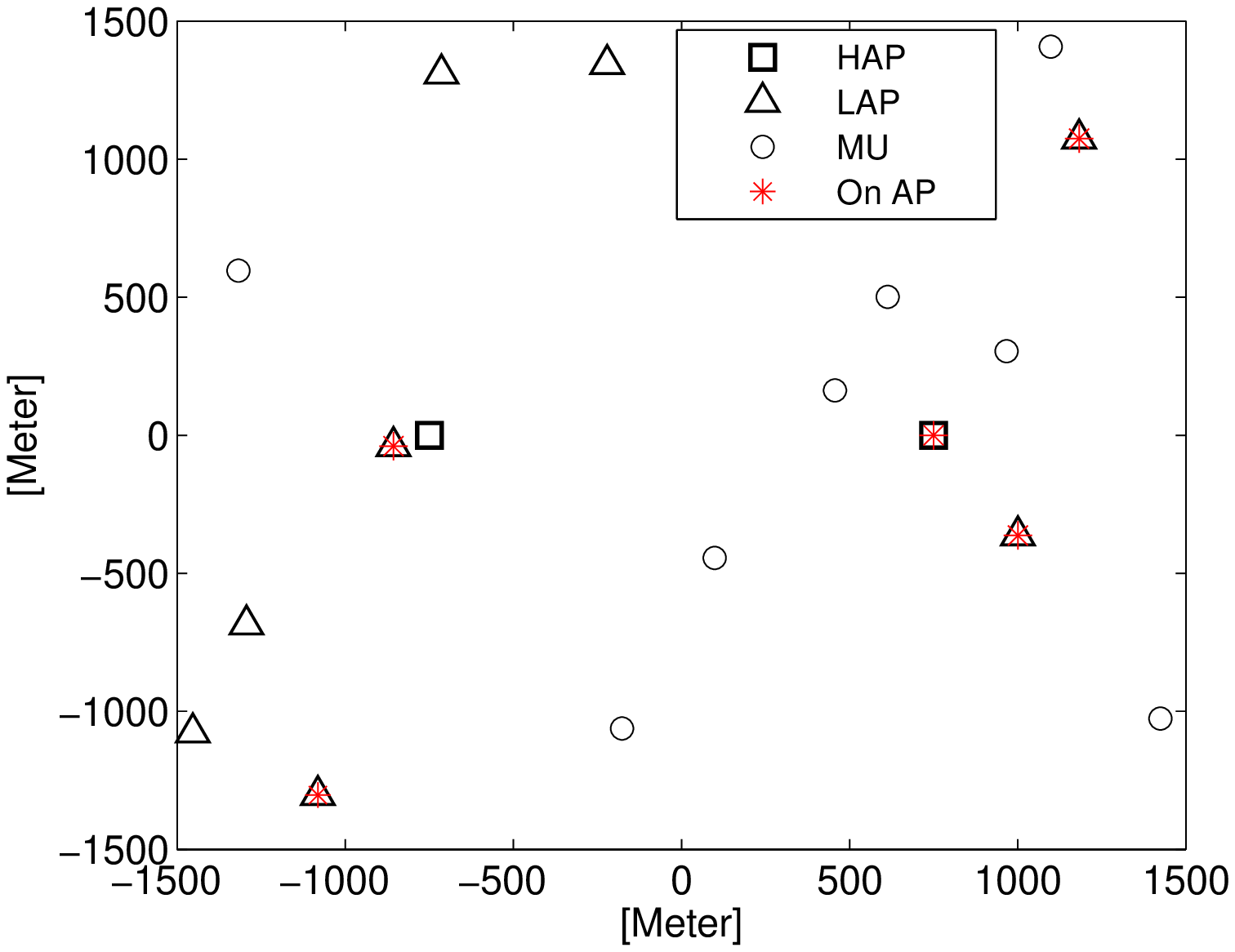}
\label{fig:FeasibilityDemoHetNetClose}
}
\subfigure[]{
\includegraphics[width=0.48\textwidth]{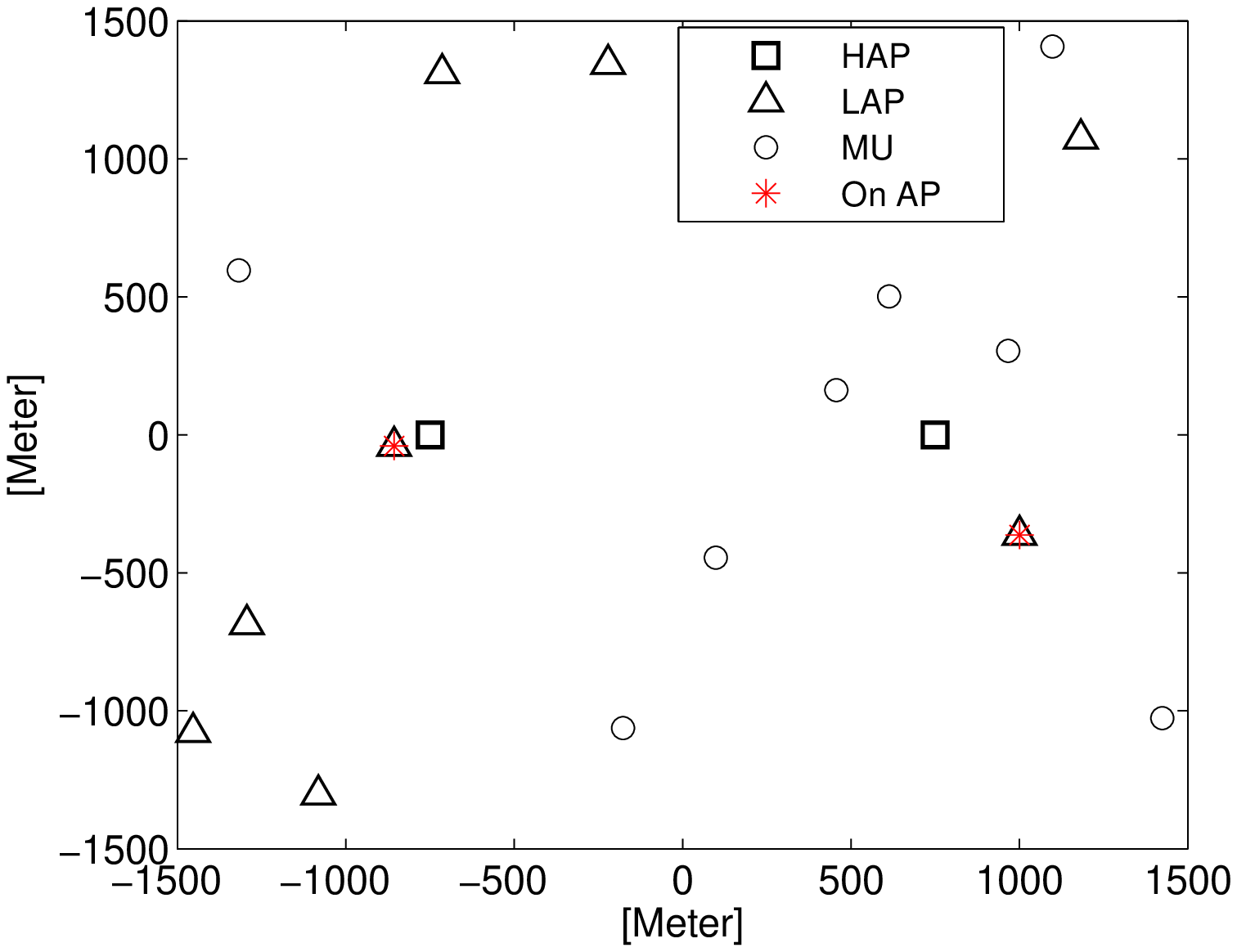}
\label{fig:FeasibilityDemoHetNetDL}
}
\caption[]{The set of active APs generated by: (a) Proposed algorithms; (b) APIRSS; (c) MUIRSS; and (d) PAw/oUL.}
\label{fig:FeasibilityDemo}
\end{figure}

To compare the feasibility performance, we run the above algorithms with different DL and UL SINR targets. It is assumed that $N = 6$ and $K = 4$. The results are summarized in Table \ref{table4} and Table \ref{table3}, where the number of infeasible cases for each scheme is shown over {\color{red}$200$ randomly generated network and channel realizations}, for homogeneous setup and heterogeneous setup, respectively. Note that in these examples, Algorithm I and Algorithm II have identical feasibility performance, since the system is infeasible only when the DL/UL SINR requirements cannot be supported for given channels and power budgets even with all APs being active.

From both Table \ref{table4} and Table \ref{table3}, it is first observed that the three comparison schemes, i.e., APIRSS based scheme, MUIRSS based scheme and PAw/oUL, all incur much larger number of infeasible cases as compared to our proposed algorithms. It is also observed that among the three comparison schemes, PAw/oUL has the best performance (or the minimum number of infeasible cases) when the DL transmission is dominant (i.e., $\gamma^{\text{DL}}_i > \gamma^{\text{UL}}_i$); however, it performs the worst in the opposite situation (i.e., $\gamma^{\text{DL}}_i < \gamma^{\text{UL}}_i$). This observation indicates that DL oriented scheme could result in infeasible transmit power of MUs in the UL for the cases with stringent UL requirements. From the last two rows of Table \ref{table3}, it is observed that the APIRSS based scheme performs worse than the MUIRSS based scheme when the UL SINR target is high. This is because that under heterogeneous setup, as shown in Fig. \ref{fig:FeasibilityDemoHetNetHighSignal}, MUs are attached to the HAPs under APIRSS based scheme although the HAPs may be actually more distant away from MUs compared with the distributed LAPs. This imbalanced association causes much higher transmit powers of MUs or even their infeasible transmit power in the UL. There has been effort in the literature to address this traffic imbalance problem in HetNet. For example in \cite{Andrews13}, the reference signal from picocell BS is multiplied by a factor with magnitude being larger than one, which makes it appear more appealing for MU association than the heavily-loaded macrocell BS.

\begin{table}[ht]
\caption{Feasibility Performance Comparison under homogeneous setup}
\begin{center}
\begin{tabular}{|c|c|c|c|c|c|c|}
  \hline
  \multicolumn{3}{|c|}{Parameters}  &  \multicolumn{4}{|c|}{Number of Infeasible Cases}\\ \hline
  $K$ & $\gamma^{\text{DL}}_i$(dB) & $\gamma^{\text{UL}}_i$(dB) & APIRSS & MUIRSS & PAw/oUL & Proposed Algorithms\\ \hline
  $2$ &$6$ & $6$  & $0$ & $0$ & $0$ & $0$  \\ \hline
  $2$ & $12$ & $6$  & $14$ & $18$ & $2$ & $0$ \\ \hline
  $2$ & $6$ & $12$  & $86$ & $90$ & $150$ & $54$ \\ \hline
  $2$ & $12$ & $12$  & $88$ & $92$ & $118$ & $50$ \\ \hline
  $4$ & $6$ & $6$  & $0$ & $0$ & $4$ & $0$ \\ \hline
  $4$ & $12$ & $6$  & $52$ & $56$ & $2$ & $0$ \\ \hline
  $4$ & $6$ & $12$  & $140$ & $144$ & $194$ & $80$ \\ \hline
  $4$ & $12$ & $12$  & $152$ & $150$ & $164$ & $86$ \\
  \hline
\end{tabular}
\end{center}\label{table4}
\end{table}

\begin{table}[ht]
\caption{Feasibility Performance Comparison under Heterogeneous setup}
\begin{center}
\begin{tabular}{|c|c|c|c|c|c|c|}
  \hline
  \multicolumn{3}{|c|}{Parameters}  &  \multicolumn{4}{|c|}{Infeasibility}\\ \hline
  $K$ & $\gamma^{\text{DL}}_i$(dB) & $\gamma^{\text{UL}}_i$(dB) & APIRSS & MUIRSS & PAw/oUL & Proposed Algorithms\\ \hline
  $2$ &$6$ & $6$  & $0$ & $0$ & $0$ & $0$  \\ \hline
  $2$ & $12$ & $6$ & $12$ & $18$ & $2$ & $0$\\ \hline
  $2$ & $6$ & $12$ & $124$ & $82$ & $154$ & $46$\\ \hline
  $2$ & $12$ & $12$ & $124$ & $86$ & $118$ & $52$\\ \hline
  $4$ & $6$ & $6$ & $0$ & $0$ & $2$ & $0$\\ \hline
  $4$ & $12$ & $6$ & $42$ & $36$ & $6$ & $0$\\ \hline
  $4$ & $6$ & $12$ & $178$ & $132$ & $196$ & $84$\\ \hline
  $4$ & $12$ & $12$ & $178$ & $134$ & $170$ & $82$\\
  \hline
\end{tabular}
\end{center}\label{table3}
\end{table}

\subsection{Sum-Power Minimization}\label{sec:power performance}
Next, we compare the performance of the proposed algorithms in terms of sum-power minimization in C-RAN with the following benchmark schemes:
\begin{itemize}
\item {\bf Exhaustive search (ES)}: In this scheme, the optimal set of active APs are found by exhaustive search, which serves as the performance upper bound (or lower bound on the sum-power consumption) for other considered schemes. With any set of active APs, the minimum-power DL and UL beamforming problems can be separately solved. Since the complexity of ES grows exponentially with $N$, it can only be implemented for C-RAN with small number of APs.

\item {\bf Joint processing (JP) among all APs}\cite{Weiyu10}: In this scheme, all the APs are assumed to be active and only the total transmit power consumption is minimized by solving two separate (DL and UL) minimum-power beamforming design problems.

\item {\bf Algorithm I with $\ell_{1,\infty}$ norm penalty}: This algorithm is the same as that given in Table \ref{table1} except that the sparsity enforcing penalty is replaced with $\ell_{1,\infty}$ norm as given in (\ref{eq:infinity}).
\end{itemize}

In our simulations, we consider the homogeneous C-RAN setup with $N = 6$ and plot the performance by averaging over {\color{red}$500$} randomly generated network and channel realizations. The SINR requirements are set as $\gamma^{\text{UL}}_i = \gamma^{\text{DL}}_i = 8$dB for all MUs $i\in \mathcal{K}$. Fig. \ref{fig:PerformanceMU} and Fig. \ref{fig:PerformanceComparePcDis} show the sum-power consumption achieved by different algorithms versus the number of MUs $K$ and AP static power consumption $P_{c,n}$ (assumed to be identical for all APs), respectively. From both figures, it is observed that the proposed algorithms have similar performance as the optimal ES and achieve significant power saving compared with JP. It is also observed that the penalty term based on either $\ell_{1,2}$ or $\ell_{1,\infty}$ norm has small impact on the performance of Algorithm I. Finally, Algorithm I always outperforms Algorithm II although the performance gap is not significant.

\begin{figure}
\centering
\epsfxsize=0.7\linewidth
\includegraphics[width=12cm]{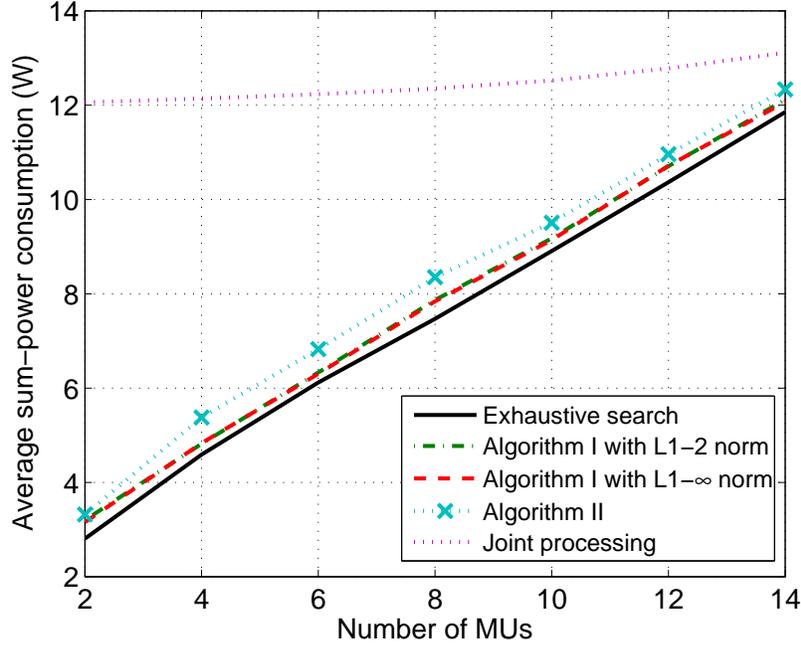}
\caption{Sum-power consumption versus number of MUs under homogeneous setup with $P_{c,n} = 2$W, $\forall n \in \mathcal{N}$.}
\label{fig:PerformanceMU}
\end{figure}

\begin{figure}
\centering
\epsfxsize=0.7\linewidth
    \includegraphics[width=12cm]{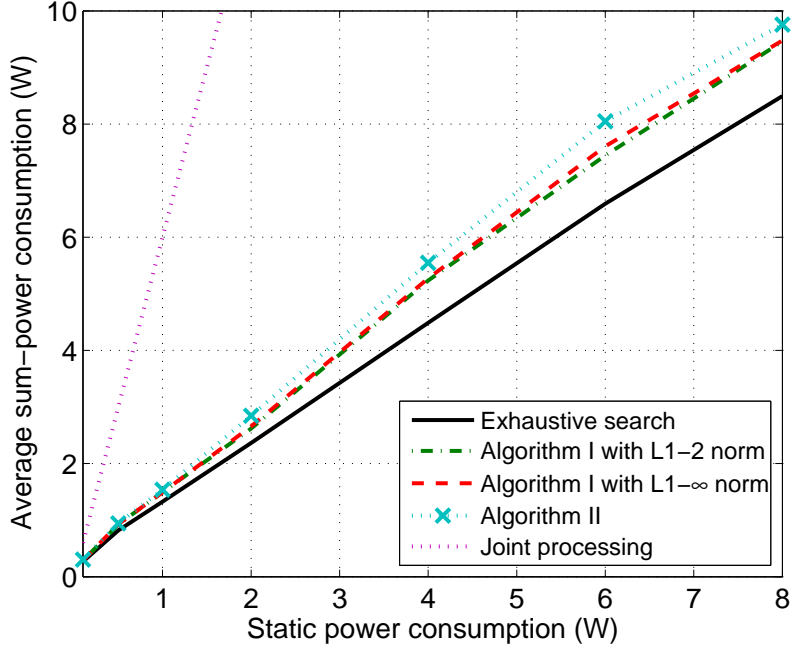}
\caption{Sum-power consumption versus AP static power consumption under homogeneous setup with $K = 4$.} \label{fig:PerformanceComparePcDis}
\end{figure}

\subsection{Power Consumption Tradeoff}
Finally, we compare the sum-power consumption tradeoffs between active APs and all MUs for the proposed algorithms as well as the optimal ES, by varying the weight parameter $\lambda$ in our formulated problems. We consider a homogenous C-RAN setup with $N = 6$ and $K = 4$, where $\gamma^{\text{UL}}_i = \gamma^{\text{DL}}_i = 8$dB for all MUs $i\in \mathcal{K}$. Since it has been shown in the pervious subsection that Algorithm I with $\ell_{1,2}$ norm and $\ell_{1,\infty}$ norm achieves similar performance, we choose $\ell_{1,2}$ norm in this simulation. Furthermore, since JS assumes that all the APs are active, which decouples DL and UL transmissions and thus has no sum-power consumption tradeoffs between APs and MUs, it is also not included. From Fig. \ref{fig:Tradeoff}, it is first observed that for all considered algorithms, as $\lambda$ increases, the sum-power consumption of active APs increases and that of all MUs decreases, which is as expected. It is also observed that Algorithm I achieves trade-off performance closer to ES and outperforms Algorithm II, which is in accordance with the results in Figs. \ref{fig:PerformanceMU} and \ref{fig:PerformanceComparePcDis}.

\begin{figure}
\centering
\epsfxsize=0.7\linewidth
\includegraphics[width=12cm]{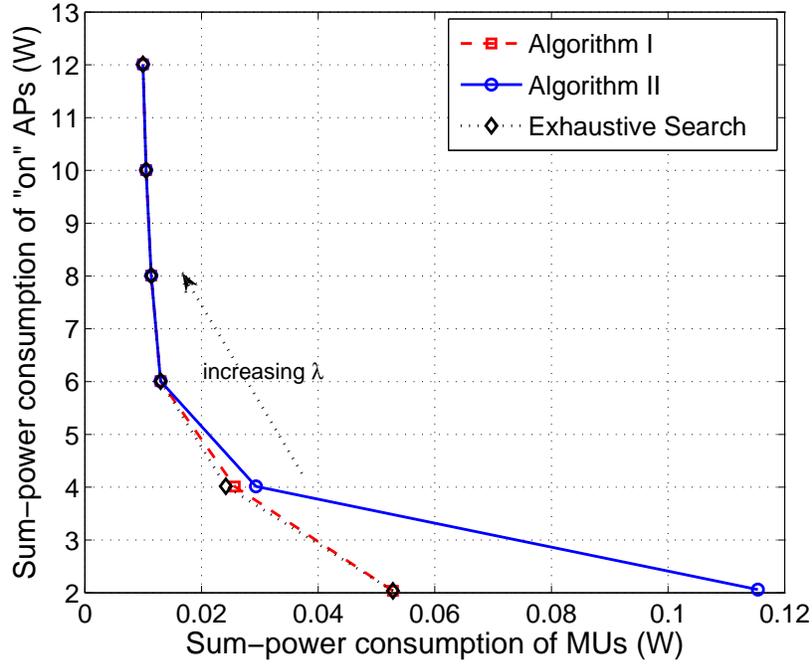}
\caption{Sum-power consumption tradeoffs between active APs and MUs under homogeneous setup.}
\label{fig:Tradeoff}
\end{figure}

\section{Conclusion}\label{sec:conclusion}
In this paper, we consider C-RAN with densely deployed APs cooperatively serving distributed MUs for both the DL and UL transmissions. We study the problem of joint DL and UL MU-AP association and beamforming design to optimize the energy consumption tradeoffs between the active APs and MUs.
Leveraging on the celebrated UL-DL duality result, we show that by establishing a virtual DL transmission for the original UL transmission, the joint DL and UL problem can be converted to an equivalent DL problem in C-RAN with two inter-related subproblems for the original and virtual DL transmissions, respectively. Based on this transformation, two efficient algorithms for joint DL and UL MU-AP association and beamforming design are proposed based on GSO and RIP techniques, respectively. By extensive simulations, it is shown that our proposed algorithms improve the network reliability/feasibility, energy efficiency, as well as power consumption tradeoffs between APs and MUs, as compared to other existing methods in the literature.

\end{document}